\documentclass{article}
 
\RequirePackage[margin=1.0in]{geometry} 

\RequirePackage[american]{babel}
\RequirePackage{csquotes}
\RequirePackage[T1]{fontenc}

\RequirePackage{amsmath, mathtools}
\RequirePackage{amssymb, mathrsfs, stmaryrd}

\RequirePackage[bbgreekl]{mathbbol}
\DeclareSymbolFontAlphabet{\mathbbm}{bbold}
\DeclareSymbolFontAlphabet{\mathbb}{AMSb}

\RequirePackage{amsthm}
\RequirePackage{thm-restate} 

\newcommand{\eqdef}{\stackrel{\text{\tiny def}}{=}} 

\newcommand{\from}{\colon} 

\DeclarePairedDelimiter{\xpar}{\lparen}{\rparen} 

\DeclarePairedDelimiterX{\oointerval}[2]{\lparen}{\rparen}{#1,\,#2} 
\DeclarePairedDelimiterX{\ocinterval}[2]{\lparen}{\rbrack}{#1,\,#2} 
\DeclarePairedDelimiterX{\cointerval}[2]{\lbrack}{\rparen}{#1,\,#2} 
\DeclarePairedDelimiterX{\ccinterval}[2]{\lbrack}{\rbrack}{#1,\,#2} 

\DeclarePairedDelimiter{\ceil}{\lceil}{\rceil} 

\newcommand{\aset}[1]{\mathcal{#1}}

\providecommand{\given}{}
\DeclarePairedDelimiterX{\set}[1]{\{}{\}}{%
  \renewcommand{\given}{\nonscript\;\delimsize\vert\nonscript\;\mathopen{}} #1%
} 

\DeclarePairedDelimiter{\range}{\llbracket}{\rrbracket} 

\providecommand{\given}{}
\DeclarePairedDelimiterX{\meanbrackets}[1]{\lbrack}{\rbrack}{%
  \renewcommand{\given}{\nonscript\,\delimsize\vert\nonscript\,\mathopen{}} #1%
}
\DeclareMathOperator*{\meansymbol}{\mathbb{E}}
\NewDocumentCommand{\mean}{e{_}}{%
  \meansymbol%
  \IfNoValueF{#1}{\sb{#1}}%
  \meanbrackets%
} 

\DeclarePairedDelimiterX{\probparen}[1]{\lparen}{\rparen}{%
  \renewcommand{\given}{\nonscript\,\delimsize\vert\nonscript\,\mathopen{}} #1%
}
\DeclareMathOperator*{\probsymbol}{Pr}
\NewDocumentCommand{\prob}{e{_}}{%
  \probsymbol%
  \IfNoValueF{#1}{\sb{#1}}%
  \probparen%
} 

\DeclareMathOperator{\simplex}{\Delta} 

\newcommand{\param}[1]{%
  \texorpdfstring{\ensuremath{#1}\protect\nobreakdash}{#1}%
} 

\DeclarePairedDelimiter{\abs}{\lvert}{\rvert} 

\newcommand{\amatrix}[1]{\mathbf{#1}} 

\DeclareMathOperator{\rank}{rank} 

\newcommand*{\T}{\mkern-1.5mu\mathsf{T}} 

\newcommand{\conj}[1]{\overline{#1}} 

\newcommand*{\A}{\dagger} 

\newcommand*{\inv}{-1} 

\DeclareMathOperator{\lin}{span} 

\DeclarePairedDelimiterX{\inner}[2]{\langle}{\rangle}{#1,\,#2} 
\DeclarePairedDelimiter{\norm}{\lVert}{\rVert} 

\newcommand{\agraph}[1]{\mathscr{#1}}

\newcommand{\bigoh}{\mathcal{O}} 
\newcommand{\bigtheta}{\mathbbm{\Theta}} 

\DeclareMathOperator{\polytime}{poly} 

\newcommand{\aproblem}[1]{\mathcal{#1}} 

\DeclarePairedDelimiter{\ket}{\lvert}{\rangle} 

\DeclarePairedDelimiterX{\braket}[3]{\langle}{\rangle}{%
  #1\,\delimsize\vert\,\mathopen{}%
  \ifblank{#2}{}{#2\,\delimsize\vert\,\mathopen{}}%
  #3%
} 

\DeclarePairedDelimiterX{\ketbra}[2]{\lvert}{\rvert}{%
  #1\,\delimsize\rangle\!\delimsize\langle\,\mathopen{}%
  #2%
} 

\DeclarePairedDelimiterX{\comm}[2]{\lbrack}{\rbrack}{#1,\,#2} 

\DeclarePairedDelimiterX{\anticomm}[2]{\lbrace}{\rbrace}{#1,\,#2} 

\providecommand{%
  \naturals}{\mathbb{N}}            

\providecommand{%
  \integers}{\mathbb{Z}}            

\providecommand{%
  \nonnegativeintegers
}{%
  \integers_{\geq 0}}               

\providecommand{%
  \rationals}{\mathbb{Q}}           

\providecommand{%
  \reals}{\mathbb{R}}               

\providecommand{%
  \simplex}{\Delta}                 

\providecommand{%
  \nonnegativereals
}{%
  \reals_{\geq 0}}                  
\let\nnreals\nonnegativereals       

\providecommand{%
  \complex}{\mathbb{C}}             

\providecommand{%
  \complextorus}{\mathbb{T}}        

\providecommand{%
  \polynomials}{\reals}             

\providecommand{%
  \apolynomial}{p}                  

\providecommand{%
  \sumofsquares}{\mathbbm{\Sigma}}  

\providecommand{%
  \quadraticmodule}{Q}              

\providecommand{%
  \asumofsquares}{\sigma}           
\let\asos\asumofsquares             

\providecommand{%
  \unitvector}{e}                   

\providecommand{%
  \identitymatrix}{\amatrix{I}}     

\providecommand{%
  \hessian}{\amatrix{H}}            

\providecommand{%
  \jacobian}{\amatrix{J}}           

\providecommand{%
  \symmetrizedjacobian%
}{\amatrix{SJ}}                     

\providecommand{%
  \mineigenvalue%
}{%
  \lambda_{\mathrm{min}}%
}                                   

\providecommand{%
  \maxeigenvalue%
}{%
  \lambda_{\mathrm{max}}%
}                                   

\providecommand{%
  \im}{i}                           

\RequirePackage[algoruled, vlined]{algorithm2e}
\SetAlgoSkip{bigskip}

\RequirePackage{graphicx}
\DeclareGraphicsExtensions{.pdf, .png}

\RequirePackage{float}

\RequirePackage{tikz}
\RequirePackage{tikz-3dplot}
\usetikzlibrary{calc, positioning, angles, quotes, arrows.meta}

\RequirePackage{quantikz}
\usetikzlibrary{quantikz2}

\RequirePackage{xcolor}

\definecolor{OIblack}{RGB}{0,0,0}
\definecolor{OIorange}{RGB}{230,159,0}
\definecolor{OIskyblue}{RGB}{86,180,233}
\definecolor{OIbluegreen}{RGB}{0,158,115}
\definecolor{OIyellow}{RGB}{240,228,66}
\definecolor{OIblue}{RGB}{0,114,178}
\definecolor{OIvermillion}{RGB}{213,94,0}
\definecolor{OIredpurple}{RGB}{204,121,167}

\RequirePackage[numbers]{natbib}
\RequirePackage{hyperref}
\definecolor{MyLinkColor}{HTML}{800006}
\definecolor{MyCiteColor}{HTML}{2E7E2A}
\definecolor{MyFileColor}{HTML}{131877}
\definecolor{MyURLColor}{HTML}{8A0087}
\definecolor{MyMenuColor}{HTML}{727500}
\definecolor{MyRunColor}{HTML}{137776}
\hypersetup{
  colorlinks,
  linkcolor=MyLinkColor,
  citecolor=MyCiteColor,
  filecolor=MyFileColor,
  urlcolor=MyURLColor,
  menucolor=MyMenuColor,
  runcolor=MyRunColor
} 

\RequirePackage[
  prefix,
  abbreviations,
  symbols,
  nogroupskip, 
  style = long,
  stylemods = bookindex,
  automake
]{glossaries-extra}

\GlsXtrEnableEntryCounting{abbreviation}{0} 
\GlsXtrEnableInitialTagging{abbreviation}{\initial}

\glsaddkey{dashed}{}
  {\glsentrydashed}
  {\Glsentrydashed}
  {\glsdashed}
  {\Glsdashed}
  {\GLSdashed}   

\ExplSyntaxOn 

\cs_new_protected:Npn \jglsentryprefixfirstorprefix:nn #1#2
  {\ifglshasfield{prefixfirst}{#2}%
    {\csname #1lsentryprefixfirst\endcsname{#2}}%
    {\csname #1lsentryprefix\endcsname{#2}}}%
\newcommand*{\glsentryprefixfirstorprefix}[1]
  {\jglsentryprefixfirstorprefix:nn {g}{#1}}
\newcommand*{\Glsentryprefixfirstorprefix}[1]
  {\jglsentryprefixfirstorprefix:nn {G}{#1}}
\glsmfuaddmap{\glsentryprefixfirstorprefix}{\Glsentryprefixfirstorprefix}

\cs_new_protected:Npn \newjgls:nnnnn #1#2#3#4#5 {%
  \newjglsstarred:nnnn {#1}{#3}{#4}{#5}%
  \newjglsstarred:nnnn {#2}{#3}{#4}{#5}%
  \exp_args:Ncc \glsmfuaddmap {#1#4}{#2#4}%
}
\cs_new_protected:Npn \newjglsstarred:nnnn #1#2#3#4
  {\expandafter\NewDocumentCommand%
    \csname #2#1#3\endcsname {s O{} m O{}}%
    {\IfBooleanTF{##1}%
      {#4{#1}{*}{##2}{##3}{##4}}%
      {#4{#1}{}{##2}{##3}{##4}}}}

\cs_new_protected:Npn \jglsdashedorshort:nnnnn #1#2#3#4#5
  {\csname glslink#2\endcsname[#3]{#4}{\ifglsused{#4}%
    {\csname #1lsentryshort\endcsname{#4}#5}%
    {\csname #1lsentrydashed\endcsname{#4}#5~(\csname glsentryshort#2\endcsname{#4})}\glsunset{#4}}}
\newjgls:nnnnn {g}{G}{}{lsdashedorshort}{\jglsdashedorshort:nnnnn}
 
\cs_new_protected:Npn \jpglsdashedorshort:nnnnn #1#2#3#4#5
  {\ifglsused{#4}
    {\csname #1lsentryprefix\endcsname{#4}}
    {\csname #1lsentryprefixfirstorprefix\endcsname{#4}}%
  \csname glslink#2\endcsname[#3]{#4}{\ifglsused{#4}%
    {\glsentryshort{#4}#5}%
    {\glsentrydashed{#4}#5~(\csname glsentryshort#2\endcsname{#4})}\glsunset{#4}}}
\newjgls:nnnnn {g}{G}{p}{lsdashedorshort}{\jpglsdashedorshort:nnnnn}

\cs_new_protected:Npn \jpglsdashed:nnnnn #1#2#3#4#5
  {\ifglsused{#4}%
    {\csname #1glsentryprefix\endcsname{#4}}%
    {\csname #1lsentryprefixfirstorprefix\endcsname{#4}\glsunset{#4}}%
  \csname glsdashed#2\endcsname[#3]{#4}[#5]}
\newjgls:nnnnn {p}{P}{}{glsdashed}{\jpglsdashed:nnnnn}

\cs_new_protected:Npn \jdpgls:nnnnn #1#2#3#4#5
  {\ifglsused{dual.#4}%
    {\csname #1lsfield\endcsname{#4}{prefix}}%
    {\csname #1lsfield\endcsname{#4}{prefixfirst}\glsunset{dual.#4}}%
  \csname dgls#2\endcsname[#3]{#4}[#5]}
\newjgls:nnnnn {g}{G}{dp}{ls}{\jdpgls:nnnnn}

\ExplSyntaxOff

\RequirePackage[
  noabbrev,
  capitalise,
  nameinlink
]{cleveref}
\RequirePackage{crossreftools}

\theoremstyle{plain}
\newtheorem{theorem}{Theorem}
\newtheorem{lemma}[theorem]{Lemma}

\newtheorem{corollary}[theorem]{Corollary}

\theoremstyle{definition}
\newtheorem{definition}[theorem]{Definition}
\newtheorem{assumption}{Assumption}

\theoremstyle{remark}

\Crefname{assumption}{Assumption}{Assumptions}
\Crefname{claim}{Claim}{Claims}
\Crefname{question}{Question}{Questions}
\Crefname{statement}{Statement}{Statements}

\RequirePackage[inline]{enumitem} 
\setlist[enumerate, 1]{label={\arabic*)}}
\setlist[enumerate, 2]{label={(\alph*)}}

\RequirePackage{booktabs}

\RequirePackage{siunitx}
\RequirePackage{quoting}
\newenvironment{emphasized}{%
    \begin{quoting}[font={itshape}]
}{%
    \end{quoting}
} 

\newlength{\jobjectivesensewidth}
\newlength{\jtemplength}
\providecommand{\objectivesense}{}
\NewDocumentCommand{\setobjectivesense}{m O{}}{%
  \renewcommand{\objectivesense}{#1\IfNoValueF{#2}{\sb{#2}}}%
  \setlength{\jobjectivesensewidth}
    {\widthof{\ensuremath{\displaystyle\objectivesense}}}%
  \setlength{\jtemplength}
    {\widthof{\ensuremath{\displaystyle\max}}}%
  \setlength{\jobjectivesensewidth}
    {0.5\jobjectivesensewidth + 0.5\jtemplength}}
\newcommand{\subjectto}{%
  \mathmakebox[\jobjectivesensewidth][r]{\text{s.t.}}}

\AtBeginDocument{%
\begingroup
\setbox0=\hbox{%

}%
\endgroup
}

\newabbreviation[
  prefix      = {a~},
  description = {\initial{D}iscrete \initial{F}ourier \initial{T}ransform}
]{DFT}{DFT}{Discrete Fourier Transform}

\newabbreviation[
  prefix      = {an\space},
  prefixfirst = {a~},
  description = {\initial{F}ast \initial{F}ourier \initial{T}ransform}
]{FFT}{FFT}{Fast Fourier Transform}

\newabbreviation[
  prefix      = {an\space},
  prefixfirst = {a~},
  description = {\initial{F}ully \initial{P}olynomial \initial{R}andomized \initial{A}pproximation \initial{S}cheme}
]{FPRAS}{FPRAS}{fully polynomial randomized approximation scheme}

\newabbreviation[
  prefix      = {an\space},
  prefixfirst = {a~},
  description = {\initial{H}ermitian \initial{T}rigonometric \initial{P}olynomial}
]{HTP}{HTP}{Hermitian trigonometric polynomial}

\newabbreviation[
  prefix      = {a~},
  description = {\initial{N}oisy \initial{I}ntermediate-\initial{S}cale \initial{Q}uantum Device}
]{NISQ}{NISQ}{noisy intermediate-scale quantum}

\newabbreviation[
  prefix      = {a~},
  description = {\initial{P}arametrized \initial{Q}uantum \initial{C}ircuit}
]{PQC}{PQC}{parametrized quantum circuit}

\newabbreviation[
  shortplural = {PSDs},
  longplural  = {positive semidefinite matrices},
  dashed      = {positive-semidefinite},
  prefix      = {an\space},
  prefixfirst = {a~},
  description = {\initial{P}ositive \initial{S}emi-\initial{D}efinite Matrix}
]{PSD}{PSD}{positive semidefinite}

\newabbreviation[
  prefix      = {a~},
  description = {\initial{Q}uantum \initial{A}pproximate \initial{O}ptimization \initial{A}lgorithm}
]{QAOA}{QAOA}{Quantum Approximate Optimization Algorithm}

\newabbreviation[
  shortplural = {SOS},
  longplural  = {sums of squares},
  dashed      = {sum-of-squares},
  prefix      = {an\space},
  prefixfirst = {a~},
  description = {\initial{S}um \initial{o}f \initial{S}quares}
]{SOS}{SOS}{sum of squares}

\newabbreviation[
  prefix      = {a~},
  description = {\initial{V}ariational \initial{Q}uantum \initial{A}lgorithm}
]{VQA}{VQA}{variational quantum algorithm}

\newabbreviation[
  prefix      = {a~},
  description = {\initial{V}ariational \initial{Q}uantum \initial{E}igensolver}
]{VQE}{VQE}{variational quantum eigensolver}

\newabbreviation[
  shortplural = {HEA},
  longplural  = {hardware-efficient ans{\"a}tze},
  dashed      = {hardware-efficient-ansatz},
  prefix      = {an\space},
  prefixfirst = {a~},
  description = {\initial{H}ardware-\initial{E}fficient \initial{A}nsatz}
]{HEA}{HEA}{hardware-efficient ansatz}

\glsxtrnewsymbol[
  description = {bounded-error quantum polynomial time},
  prefix      = {a~}
]{BQP}{\ensuremath{\mathsf{BQP}}}

\glsxtrnewsymbol[
  description = {bounded-error probabilistic polynomial time with access to a BQP oracle},
  prefix      = {a~}
]{BPPBQP}{\ensuremath{\mathsf{BPP}^{\mathsf{BQP}}}}

\glsxtrnewsymbol[
  description = {class of promise problems induced by \textsc{Max-Cut} problems},
  prefix      = {a~}
]{GAPMAXCUT}{\ensuremath{\mathrm{GAP\text{-}MAXCUT}}}

\glsxtrnewsymbol[
  description = {class of promise problems induced by PQC optimization problems},
  prefix      = {a~}
]{GAPPQCO}{\ensuremath{\mathrm{GAP\text{-}PQCO}}}

\glsxtrnewsymbol[
  description = {nondeterministic polynomial time},
  prefix      = {an\space}
]{NP}{\ensuremath{\mathsf{NP}}}

\glsxtrnewsymbol[
  description = {class of PQC optimization problems},
  prefix      = {a~}
]{PQCO}{\ensuremath{\mathrm{PQCO}}}

\makeglossaries
\glsdisablehyper

\RequirePackage{authblk}
\RequirePackage{lineno}  

\title{Global Optimization for Parametrized Quantum Circuits}

\author[1]{Iosif Sakos}
\author[1,3,5]{Antonios Varvitsiotis}
\author[2,4,5]{Georgios Korpas}
\author[1]{Wayne Lin}

\affil[1]{Singapore University of Technology and Design, Singapore}
\affil[2]{Quantum Technologies Group, Innovation \& Ventures, HSBC, Singapore}
\affil[3]{Centre for Quantum Technologies, Singapore}
\affil[4]{Czech Technical University in Prague, Czechia}
\affil[5]{Archimedes Research Unit on AI, Marousi, Greece}

\begin{document}

\pagenumbering{roman}
\maketitle
\thispagestyle{empty}

\begin{abstract}

In the absence of error correction, \glsxtrlong{NISQ} devices are typically operated by training \glspl{PQC} so as to minimize a suitable loss function.
Finding the optimal parameters of those circuits is a hard optimization problem, where global guarantees are known only for highly structured cases of limited practical relevance, and, more broadly, zero- and first-order methods can fail to find even local minima due to the presence of barren plateaus.
In this work, we study the training of practical classes of \glspl{PQC}, namely polynomial-depth circuits with a constant number of trainable parameters. 
This captures widely used \gls{PQC} families, including fixed-depth \glsxtrshort{QAOA}, hardware-efficient ansätze with a controlled number of parameters per layer, and Fixed Parameter Count \glsxtrshort{QAOA}. 
Our main technical result is \pgls{FPRAS}, which, for every $\epsilon > 0$, returns an \param{\epsilon}\nobreakdash-approximate solution to the problem's global optimum with high probability, and has runtime and query complexity polynomial in $1 /\epsilon$ and the number of qubits. 
Unlike the standard hybrid quantum--classical training loop in variational algorithms, where the quantum device is queried repeatedly throughout the training, our approach separates the computation into two distinct stages:
\begin{enumerate*}[label=(\arabic*)]
  \item an initial quantum data-acquisition phase, followed by 
  \item a classical global-optimization phase based on the trigonometric moment/\glsdashed{SOS} hierarchies of semidefinite programs.
\end{enumerate*}
Under a standard flat-extension condition, which can be checked numerically, the method also supports the extraction of globally optimal circuit parameters.
The existence of \pgls{FPRAS} implies that the promise problem associated with the optimization of poly-depth constant-parameter \gls{PQC} is in \gls{BQP}. 
This imposes a limitation on the expressive power of the class, namely, it cannot encode combinatorial optimization problems whose objective values are separated by an inverse-polynomial gap.
\end{abstract}
\clearpage

\tableofcontents
\clearpage

\pagenumbering{arabic}
\setcounter{page}{1}
\glsresetall 

\section{Introduction}
\label{sec:Introduction}

Quantum computing has emerged as a revolutionary computational paradigm, exploiting the properties of quantum mechanics to perform tasks that are believed to be intractable for classical computers. 
Early theoretical breakthroughs, such as Shor's factoring algorithm~\citep{shor_1994_algorithms}, demonstrated the possibility of exponential quantum speedups and motivated intensive efforts in quantum hardware development. 
These efforts have culminated in the advent of \gls{NISQ} devices~\citep{preskill_2018_quantum}. 
These are digitally controlled quantum processors that operate without full fault-tolerant error correction and are therefore constrained by noise and finite coherence times.

\Pgls{NISQ} device, commonly, executes a sequence of fixed and parametrized unitary gates, giving rise to \pgls{PQC}. 
The choice of gates, qubit connectivity, and overall circuit architecture is typically dictated both by the available hardware and by the problem structure, and is commonly referred to as \emph{the ansatz}.
The parametrized gates depend on classical control variables $\theta \in \reals^{M}$, or in practice on a digitized domain $\mathbb{D}^M$; for example, a gate $\amatrix{U}_{i}(\theta_{i})$ may represent a single-qubit rotation by angle $\theta_{i}$ about a fixed Bloch-sphere axis.
Starting from an easily preparable reference state such as the basis state $\ket{0} \eqdef \ket{0}^{\otimes n}$, the \gls{PQC} is applied to prepare a state that encodes a solution to a computational task of interest.

Formally, let $\aset{H} \eqdef \complex^{2^{n}}$ denote the Hilbert space of an \param{n}-qubit quantum system, and $\aset{U}$ denote the set of unitary operators on $\aset{H}$.
\Pgls{PQC} $\amatrix{U} \from \reals^{M} \to \aset{U}$ with $M$ \emph{independent parameters} $\theta \in \reals^{M}$ is defined as
\begin{equation}
\label{eq:PQC}
  \amatrix{U}(\theta)
    = \amatrix{U}_{K}\xpar[\big]{\theta_{j_{K}}} \amatrix{C}_{K} \cdots \amatrix{U}_{1}\xpar[\big]{\theta_{j_{1}}} \amatrix{C}_{1},
  \qquad
  \amatrix{U}_{k}(\theta_{j_{k}})
    = \exp\xpar[\big]{-\im \theta_{j_{k}} \amatrix{V}_{k}},
    \quad \forall k \in \range{K}.
\end{equation}
Here, $K$ is the number of \emph{parametrized quantum gates}, and although, eventually, we assume that the number of trainable parameters $M$ is constant, we allow for a polynomially large $K$.
The Hermitian operators $\amatrix{V}_{1}$, \dots, $\amatrix{V}_{K} \from \aset{H} \to \aset{H}$ are called the \emph{generators} of the parametrized gates, while the operators $\amatrix{C}_{1}$, \dots, $\amatrix{C}_{K} \in \aset{U}$ are \emph{fixed quantum gates} of the circuit. 
Throughout, we refer to $\theta_{1}$, \dots, $\theta_{M} \in \reals$ as the \emph{independent (trainable) parameters} of the \gls{PQC}, and to the quantum gates $\amatrix{U}_{1}$, \dots, $\amatrix{U}_{K} \from \reals \to \aset{U}$ as the \emph{parametrized gates}. 
The indices $j_{1}$, \dots, $j_{K} \in \range{M}$ indicate which independent parameter $\theta_{j_{k}}$ parametrizes $\amatrix{U}_{k}$ for each $k \in \range{K}$.

In terms of tasks of interest, the majority of the literature~\citep{cerezo_2021_variational,bharti_2022_noisy,abbas_2024_challenges} revolves around two main motivating applications:
\begin{enumerate*}[label=(\roman*)]
  \item\emph{Ground-state problems}, which seek the minimum eigenvalue of a Hamiltonian $\amatrix{H}$, and are central to quantum chemistry and materials science applications.
  \item\emph{Combinatorial optimization problems}, which seek the minimum value of an objective $C$ over the hypercube $\set{0, 1}^{n}$. 
\end{enumerate*}
The starting point for both applications is \emph{the eigenvalue minimization problem}.
Given a Hermitian observable $\amatrix{O} \from \aset{H}\to \aset{H}$, we seek its minimum eigenvalue, defined variationally by
\begin{equation}
\setobjectivesense{\min}[\ket{\psi} \in \aset{H}]%
\begin{aligned}
  \mineigenvalue(\amatrix{O})
    &\eqdef \begin{aligned}[t]
      &\objectivesense%
        && \braket{\psi}{\amatrix{O}}{\psi} \\
      &\subjectto%
        &&\braket{\psi}{}{\psi} 
          = 1.
    \end{aligned}
\end{aligned}
\end{equation}
If $\amatrix{O}$ is a Hamiltonian $\amatrix{H}$ of a quantum system, then the minimum eigenvalue $\mineigenvalue(\amatrix{H})$ corresponds to the ground-state energy of the system, and the corresponding eigenvectors are the ground states; see, e.g., \citet{sakurai_2017_modern}.

Besides its fundamental physical importance, the eigenvalue minimization problem is expressive enough to also capture \gls{NP}\nobreakdash-hard discrete optimization problems for appropriate choices of the observable.
Specifically, consider an optimization problem
\begin{equation}
\label{eq:CombinatorialObjective}
  C_{\min}
    \eqdef \min_{x \in \set{0, 1}^{n}} C(x),
\end{equation}
with objective $C \from \set{0, 1}^{n} \to \reals$.
We can express \eqref{eq:CombinatorialObjective} as a combinatorial optimization problem
\begin{equation}
  C_{\min}
    = \min_{x \in \set{0, 1}^{n}} \braket{x}{\amatrix{H}_{C}}{x},
  \qquad
  \amatrix{H}_{C}
    \eqdef \sum_{x \in \set{0, 1}^{n}} C(x) \ketbra{x}{x},
\end{equation}
where the diagonal matrix $\amatrix{H}_{C}$ is called the \emph{cost Hamiltonian} of the problem.
Since $\amatrix{H}_{C}$ is diagonal, the minimum value of $C$ is equal to the minimum eigenvalue of $\amatrix{H}_{C}$, and the above formulation is a special case of the Hamiltonian optimization problem, i.e.,
\begin{equation}
  C_{\min}
    = \mineigenvalue(\amatrix{H}_{C})
    = \min_{\substack{\braket{\psi}{}{\psi} = 1}} \braket{\psi}{\amatrix{H}_{C}}{\psi}.
\end{equation}
A canonical example is the \textsc{Max-Cut} problem, where a Boolean vector $x \in \set{0, 1}^{n}$ encodes a cut in a graph, i.e., a partition of the node set into two parts, and $C(x)$ is the number of edges that cross the cut, namely those whose endpoints lie on different sides of the partition. 
In this setting we define
\begin{equation}
\label{eq:MaxCutHamiltonian}
  \amatrix{H}_{\text{MC}}
    \eqdef \sum_{(u, v) \in \aset{E}} \frac{1}{2}\xpar{\identitymatrix - \amatrix{Z}_{u} \amatrix{Z}_{v}},
\end{equation}
where $\amatrix{Z}_{u}$ is the Pauli-Z operator acting on qubit $u$, and the fact that this captures \textsc{Max-Cut} follows from the observation that 
\begin{equation}
  \frac{1}{2} \xpar{\identitymatrix - \amatrix{Z}_{u} \amatrix{Z}_{v}} \ket{x} 
    = \frac{1}{2} \xpar[\big]{1 - (-1)^{x_{u} + x_{v}}} \ket{x}.
\end{equation}

In both applications, namely the ground-state computation and combinatorial optimization, the main difficulty is that the underlying Hilbert space is exponentially large, i.e., $\aset{H} = \complex^{2^{n}}$. 
A standard way to cope with this is to use \pgls{PQC} defined by a problem-specific and/or hardware-aware ansatz, thereby restricting the search to the subset of quantum states reachable from an easily preparable reference state, such as $\ket{0}$. 
In this way, one avoids optimizing over the full Hilbert space and instead searches over a structured, parametrized family of states. 
This gives rise to the \gls{PQC} optimization~(\gls{PQCO}) problem, which seeks the minimum value
\begin{equation}\tag{\gls{PQCO}}
\label{eq:PQCO}
  f^{\star} 
    \eqdef \min_{\theta \in \reals^{M}} f(\theta),
  \qquad
  f(\theta)
    = \braket[\big]{0}{\amatrix{U}^{\A}(\theta) \amatrix{O} \amatrix{U}(\theta)}{0},
\end{equation}
where $f(\theta)$ is the expectation of the observable $\amatrix{O}$ under the state prepared by the \gls{PQC} with parameters $\theta$.
Clearly, we have $f^{\star} \geq \mineigenvalue(\amatrix{O})$, and the hope when using \glspl{PQC} to solve such problems is that: 
\begin{enumerate*}[label=(\arabic*)]
  \item the ansatz is expressive enough to provide a good approximation to the optimal value; and
  \item the optimization landscape of $f(\theta)$ is benign enough to allow for efficient optimization.
\end{enumerate*}
 
However, in the absence of full quantum error correction, only relatively shallow circuits can be executed on \gls{NISQ} hardware with meaningful reliability.
The predominant strategy for leveraging such devices despite noise and coherence limitations is to adapt the circuit parameters through iterative, task-driven optimization.
This is typically implemented in \emph{classical--quantum hybrid optimization loops} known as \glspl{VQA}; see, e.g., \citet{cerezo_2021_variational, bharti_2022_noisy} and references therein.
For a given parameter choice $\theta$, the quantum device prepares the state $\amatrix{U}(\theta)\ket{0}$ that is then used to estimate the value of $f(\theta)$.
Subsequently, a classical optimizer, such as stochastic gradient descent, updates $\theta$ so as to decrease $f(\theta)$.
This procedure is repeated until a chosen termination criterion (e.g., a sufficiently small gradient norm) is satisfied.

In the special case where the observable is a cost Hamiltonian $\amatrix{H}_{C}$, the training process admits a natural interpretation in terms of probability distributions over candidate solutions. 
Starting from an initial distribution, e.g., the uniform superposition $\frac{1}{\sqrt{2^{n}}} \sum_{x \in \set{0, 1}^{n}} \ket{x}$, the parameterized quantum circuit $\amatrix{U}(\theta)$ defines, for every parameter vector $\theta$, a probability distribution $p_{\theta}$ over Boolean vectors $x \in \set{0, 1}^{n}$.
Then, for each $p_{\theta}$, the expected cost is
\begin{equation}
  \mean_{x \sim p_{\theta}}[\big]{C(x)}
    = \sum_{x \in \set{0,1}^{n}} p_{\theta}(x)\, C(x),
\end{equation}
and training adjusts $\theta$ so as to progressively concentrate probability mass on better solutions.
While the parameter updates are carried out classically, the quantum device is used to estimate the expectation of the cost function under the exponentially large distribution $p_{\theta}$, and solve the optimization problem
\begin{equation}
  f^{\star}
    = \min_{\theta \in \reals^{M}} \mean_{x \sim p_{\theta}}[\big]{C(x)},
\end{equation}

From an algorithmic perspective, the \gls{PQC} optimization problem is challenging. 
On the one hand, it is strongly \gls{NP}\nobreakdash-hard in general, even under favorable assumptions, such as access to quantum expectation values via a quantum oracle, classical simulability of the objective (e.g., logarithmic-depth circuits), or even restriction to depth-one circuits~\citep{bittel_2021_training}. 
On the other hand, even computing a local minimum with gradient-based methods is often difficult in practice, because the objective landscape over the parameters $\theta$ may exhibit \emph{barren plateaus}, i.e., regions where the function is nearly flat and gradients carry essentially no useful optimization information; see, e.g., \citet{mcclean_2018_barren,cerezo_2025_provable} and references therein.

In view of this, several works have focused on specific subclasses of ansätze for which training can be shown to converge globally.
Existing results are largely limited either to very low-dimensional settings ($M = 1$)~\citep{ wang_2018_quantum, sureshbabu_2024_parameter} or to highly structured regimes~\citep{you_2023_convergence, wiedmann_2025_convergence}, and therefore do not fully capture practical settings of interest.
In particular, \citet{you_2023_convergence} show that global convergence of the \gls{VQE} can be guaranteed with high probability when the number of independent parameters $M$ exceeds a threshold that scales polynomially with the Hilbert-space dimension, i.e., as $\bigoh\xpar[\big]{\polytime(2^{n})}$.
\Citet{wiedmann_2025_convergence} further show that \gls{VQE} converges to the global optimum with high probability under a local surjectivity assumption on the ansatz, a condition that is quite restrictive and generally does not hold for practical ansätze.
Finally, \citet{wang_2018_quantum, sureshbabu_2024_parameter} derive analytic solutions to the \gls{PQCO} problem of the \gls{QAOA} ansatz at circuit depth $K = 1$ for the \textsc{Max-Cut} problem.

\paragraph{Contributions.} In this work, we depart from the restrictive settings considered in the existing global optimization literature for the \gls{PQC} optimization problem, and study a class of ansätze well aligned with practical implementations.
In particular, we consider ansätze with a constant (with respect to the number of qubits $n$) number of independent parameters $M$, but a polynomially large number of parametrized gates $K$; we refer to such circuits as \emph{poly-depth constant-parameter} \glspl{PQC}.
Our driving questions are:
\begin{emphasized}
  \begin{enumerate}[label=\textbf{Q\arabic*}]
    \item\label[question]{question:GlobalOptimality} Can we design methods that solve the \gls{PQCO} problem in this setting to global optimality?
    \item\label[question]{question:Expressivity} Are these ansätze expressive enough to capture hard problems, even approximately?
  \end{enumerate}
\end{emphasized}

We answer both questions by means of our main technical result, namely that, for poly-depth, constant-parameter \glspl{PQC}, and under mild assumptions on the observable and the generators, the \gls{PQCO} problem admits a range-scaled additive \gls{FPRAS}. 
Specifically, for every $\epsilon, \delta > 0$, there exists an algorithm whose runtime and query complexity are polynomial in $n$, $1 / \epsilon$, and $\log(1 / \delta)$, and which returns an estimate $\hat{f}^{\star}$ satisfying
\begin{equation}
  \prob[\Big]{\abs{\hat{f}^{\star} - f^{\star}} \leq \epsilon \xpar[\big]{\max_{\theta \in \reals^{M}} f(\theta) - f^{\star}}}
    \geq 1 - \delta.
\end{equation}
As a consequence, we obtain a strong positive answer to~\ref{question:GlobalOptimality}. 
At the same time, the existence of this approximation scheme imposes strong restrictions on the expressivity of poly-depth constant-parameter \glspl{PQC}.
Specifically, it implies that the gap-promise problem corresponding to this family of \glspl{PQC} is in the complexity class \gls{BQP}, i.e., the family of promise problems solvable in polynomial time in quantum computer, and thus, no hard promise problem can be reduced to it, unless $\gls{NP} \subseteq \gls{BQP}$; an inclusion widely believed not to be true.
This yields a corresponding negative answer to~\ref{question:Expressivity}. 

To prove our main technical result, namely the existence of \pgls{FPRAS}, we proceed in three steps. 
First, we show that the objective $f(\theta)$ can be formulated as \pgls{HTP}, i.e., a function of the form
\begin{equation}
  f(\theta)
    = \sum_{\alpha \in \integers^{M}} f_{\alpha}\, \exp\xpar[\big]{\im \inner{\alpha}{\theta}},
  \qquad \theta \in \ccinterval{-\pi}{\pi}^{M},
  \qquad \conj{f_{\alpha}} = f_{-\alpha},
\end{equation}
thereby casting the \gls{PQC} optimization problem as a trigonometric optimization problem, a class of problems that has been extensively studied in the optimization literature; see, e.g., \citep{dumitrescu_2007_positive,josz_2018_lasserre,schuld_2021_effect,fontana_2022_efficient,bach_2023_exponential} and references therein.
Although this representation has been derived in special cases in the literature, to the best of our knowledge the required degree bounds for further analysis have not been derived before for poly-depth constant-parameter \glspl{PQC}.
Second, we approximate the \gls{HTP} formulation of the \gls{PQCO} using the quantum hardware to compute estimates of $f(\theta)$ on a uniform sampling grid of the parameter space, and the multidimensional \gls{FFT} to compute the Fourier coefficients of $f(\theta)$ from these estimates. 
Finally, we apply the complex trigonometric moment/\glsdashedorshort{SOS} hierarchy~\citep{dumitrescu_2007_positive,josz_2018_lasserre,bach_2023_exponential} to the approximated \gls{PQCO}, obtaining a convergent sequence of monotone lower bounds that approach the optimum in the limit, with each level computable via semidefinite programming. 

Unlike the standard hybrid quantum--classical training loop in \glspl{VQA}, where the quantum device is queried repeatedly throughout the optimization process, our approach separates the computation into two distinct stages: an initial quantum data-acquisition phase, in which the required quantum measurements are performed, followed by a purely classical global-optimization phase based on the trigonometric \gls{SOS} hierarchy.

\paragraph{List of contributions.} Concretely, in this work we make the following contributions:
\begin{enumerate}
  \item We prove a range-scaled additive \gls{FPRAS} for poly-depth constant-parameter \glspl{PQC}, which is applicable for a broad range of practical ansätze, including fixed levels of the \gls{QAOA} and certain \glspl{HEA}.
  \item We establish conditions that can be verified numerically, and under which \emph{optimal \gls{PQC} parameters} can be extracted from the solution.
  \item We analyze the expressivity of poly-depth constant-parameter \glspl{PQC} and establish that unless $\gls{NP} \subseteq \gls{BPPBQP}$, they cannot represent \gls{NP}-hard problems, even approximately.
\end{enumerate}

\paragraph{Notation.}
The following notation is used throughout this work.
Let $\naturals$ denote the set of natural numbers; $\integers$ the set of integer numbers; $\reals$ the set of real numbers; and $\complex$ the set of complex numbers.
For $n \in \naturals$, let $\range{n} = \set{1, \dots, n}$.
For $\alpha \in \complex$, let $\conj{\alpha} \in \complex$ denote the complex conjugate of $\alpha$.
For $\amatrix{A} \in \complex^{m \times n}$, let $\amatrix{A}^{\A}$ denote the conjugate transpose of $\amatrix{A}$.
For a Hermitian operator $\amatrix{A} \from \aset{H} \to \aset{H}$, let $\norm{\amatrix{A}}$ denote the operator norm of $\amatrix{A}$, i.e., $\norm{\amatrix{A}} \eqdef \sup_{\ket{\psi} \in \aset{H} \setminus \set{0}} \frac{\norm{\amatrix{A}\ket{\psi}}}{\norm{\ket{\psi}}}$, where $\norm{\ket{\psi}} = \sqrt{\braket{\psi}{}{\psi}}$ is the Euclidean norm of $\ket{\psi}$.
Then, by the Rayleigh quotient, we also have $\norm{\amatrix{A}} = \max_{i \in \range{2^{n}}} \abs{\lambda_{i}}$, where $\lambda_{i}$ are the eigenvalues of $\amatrix{A}$, and $2\norm{\amatrix{A}} \geq \maxeigenvalue(\amatrix{A}) - \mineigenvalue(\amatrix{A})$.

\subsection{Assumptions and relevance to ansätze used in practice}
\label{sec:Assumptions}

Throughout this work, we make the following \lcnamecrefs{ass:IntegerSpectralDifferences} on the architecture of the \gls{PQC} and the observable of the \gls{PQC} optimization problem.

\begin{assumption}
\label{ass:IntegerSpectralDifferences}
  Consider \pgls{PQC} over an \param{n}-qubit system.
  We assume that the spectra of the generators $\amatrix{V}_{1}$, \dots, $\amatrix{V}_{K}$ of the \gls{PQC} have \emph{integer differences}.
\end{assumption}

\begin{assumption}
\label{ass:PQCComplexity}
  Consider \pgls{PQC} over an \param{n}-qubit system.
  We assume that:
  \begin{enumerate}[ref=\theassumption.\arabic*]
    \item\label[assumption]{ass:IndependentParameterCount} The number of independent parameters $M$ is \textbf{constant} with respect to $n$.
    \item\label[assumption]{ass:ParametrizedGatesCount} The number of parametrized gates $K$ is bounded by a known polynomial function of $n$.
    \item\label[assumption]{ass:GeneratorSpectralDiameterBound} The maximum spectral diameter among the generators $\amatrix{V}_{1}$, \dots, $\amatrix{V}_{K}$ is bounded by a known polynomial function of $n$.
  \end{enumerate}
\end{assumption}

\begin{assumption}
\label{ass:ObservableComplexity}
  Consider an instance of the \gls{PQC} optimization problem over an \param{n}-qubit system.
  We assume that the operator norm of the observable is bounded by a polynomial function of $n$.
\end{assumption}

\Cref{ass:IntegerSpectralDifferences} is the structural condition used in \cref{thm:HTPRepresentation} to ensure that the \gls{PQC} objective can be represented as \pgls{HTP}, while \cref{ass:PQCComplexity,ass:ObservableComplexity} define the computational bounds used to analyze the complexity of the hybrid quantum--classical algorithm in \cref{sec:GlobalPQCOptimization} and prove the \gls{FPRAS}. 
Only one of the aforementioned \lcnamecrefs{ass:IntegerSpectralDifferences} is the main tradeoff, namely $M = \bigoh(1)$ (\cref{ass:IndependentParameterCount}), while the remaining \lcnamecrefs{ass:IntegerSpectralDifferences} are standard in \gls{PQC} literature; we discuss their real-world applicability below.

\paragraph{Applicability.} \Cref{ass:IntegerSpectralDifferences} is stated with integer spectral differences for simplicity, but it can be extended via \emph{a parameter shift}.
In particular, if, for each independent parameter $\theta_{j}$, there exists $\alpha_{j} > 0$ such that every generator $\amatrix{V}_{k}$ with $j_{k} = j$ has integer spectral differences after rescaling (i.e., $\alpha_{j}\,\amatrix{V}_{k}$ has integer spectral differences), then the reparametrization $\theta_{j} \mapsto \theta_{j} / \alpha_{j}$ yields an equivalent \gls{PQC} formulation satisfying \cref{ass:IntegerSpectralDifferences}.
Importantly, it is \emph{not required to perform} such a reparametrization; it suffices to know that it holds.
In this sense, \cref{ass:IntegerSpectralDifferences} is a reasonably benign \lcnamecref{ass:IntegerSpectralDifferences}.
Moreover, various well-known architectures, such as \gls{QAOA}, and more generally \glspl{PQC} decomposable into Pauli-string generators, satisfy \cref{ass:IntegerSpectralDifferences} explicitly (cf.~\cref{sec:FixedLevelQAOA}).

\Cref{ass:ParametrizedGatesCount}, \ref{ass:GeneratorSpectralDiameterBound}, and \ref{ass:ObservableComplexity} are standard in the literature~\citep{cerezo_2021_variational,tilly_2022_variational}.
In particular, the generators $\amatrix{V}_{1}, \dots, \amatrix{V}_{K}$ and the observable $\amatrix{O}$ are typically implemented as \emph{sums of polynomially many local norm-bounded terms}~\citep{whitfield_2011_simulation,tilly_2022_variational}, e.g.,
\begin{equation}
  \amatrix{O}
    = \sum_{i} \alpha_{i} \amatrix{O}_{i},
\end{equation}
where each $\amatrix{O}_{i}$ acts on a constant number of qubits, and $\abs{\alpha_{i}}$ is bounded by a polynomial function of $n$.
For example, in the case of the \textsc{Max-Cut} Hamiltonian $\amatrix{H}_{\text{MC}}$ in \eqref{eq:MaxCutHamiltonian} on a graph with $n$ vertices, $\amatrix{H}_{\text{MC}}$ is decomposed as
\begin{equation}
  \amatrix{H}_{\text{MC}}
    = \sum_{(u, v) \in \aset{E}} \frac{1}{2}\xpar{\identitymatrix - \amatrix{Z}_{u} \amatrix{Z}_{v}},
\end{equation}
where each term $\identitymatrix - \amatrix{Z}_{u} \amatrix{Z}_{v}$ acts on qubits $u$ and $v$, and has operator norm at most $2$, while the number of terms is $\abs{\aset{E}} \leq \abs{\aset{V}}^{2} = n^{2}$, which is polynomial in $n$ (cf.~\cref{sec:FixedLevelQAOA} for a complete example).
It is not difficult to see that, in such cases, the operator norms, and by extension the spectral diameters, of the generators $\amatrix{V}_{1}, \dots, \amatrix{V}_{K}$ and the observable $\amatrix{O}$ are bounded by the sum of operator norms of the local terms, and therefore are bounded by polynomial functions of $n$.
Finally, $K$ being polynomial in $n$ corresponds to the standard assumption that the circuit size is polynomial in $n$~\citep{bernstein_1997_quantum}.
Importantly, these \lcnamecrefs{ass:PQCComplexity}  do not presume a Pauli-string decomposition, and therefore are applicable even for \glspl{HEA} (cf.~\cref{sec:HEAWithNonInvolutoryGenerators}).

On the other hand, \cref{ass:IndependentParameterCount} is more subtle.
The size of the trigonometric moment/\glsdashedorshort{SOS} relaxation scales exponentially in $M$, so fixing $M$ is the lever that enables polynomial-time guarantees in $n$.
Regimes where $M$ grows with $n$ are not captured by this framework.
This is the key nonstandard \lcnamecref{ass:IndependentParameterCount} in this work.
Relative to prior global guarantees, which are limited to very low-dimensional $M = 1$ or otherwise highly structured regimes~\citep{nakanishi_2020_sequential, you_2023_convergence, schatzki_2024_theoretical}, \cref{ass:IndependentParameterCount} broadens the scope while keeping the problem tractable.
Importantly, a constant number of independent parameters $M$ does not imply that the circuit is shallow, since the number of gates $K$ is allowed to scale polynomially in the number of qubits $n$, and the generators $\amatrix{V}_{1}$, \dots, $\amatrix{V}_{K}$ do not, in general, commute.

\subsection{Examples of \glsfmtshort{PQC} architectures}
\label{sec:Examples}

We provide several motivating examples of \glspl{PQC} architectures that satisfy the above \lcnamecrefs{ass:IntegerSpectralDifferences}.
We begin with a fixed level of the celebrated \gls{QAOA} ansatz~\cite{farhi_2014_quantum} in \cref{sec:FixedLevelQAOA}.
Then, we extend the scope to more general \gls{PQC} architectures.
In particular, in \cref{sec:HEAWithNonInvolutoryGenerators} we provide an example of \pgls{HEA} for ion-trapped computers, which can be optimized efficiently by solving \pgls{PQC} optimization problem satisfying \cref{ass:IntegerSpectralDifferences,ass:PQCComplexity,ass:ObservableComplexity}.
These examples serve to illustrate the applicability of our \lcnamecrefs{ass:IntegerSpectralDifferences} in practical scenarios.

\subsubsection{A fixed-level \glsfmtshort{QAOA}}
\label{sec:FixedLevelQAOA}

Consider an undirected graph $\agraph{G} \equiv \agraph{G}(\aset{V},\aset{E})$ on $\abs{\aset{V}} = n$ vertices and $\abs{\aset{E}}$ edges.
The \textsc{Max-Cut} Hamiltonian of $\agraph{G}$ is given in \eqref{eq:MaxCutHamiltonian}.
The depth\nobreakdash-$p$ \gls{QAOA} ansatz is given by
\begin{equation}
  \amatrix{U}_{\mathrm{QAOA}}(\beta, \gamma)
    = \prod_{\ell = 1}^{p}
      \exp\xpar[\big]{-\im \beta_{\ell}\, \amatrix{B}}
      \exp\xpar[\big]{-\im \gamma_{\ell}\, \amatrix{H}_{\mathrm{MC}}},
  \qquad
  \amatrix{B}
    \eqdef \sum_{u \in \aset{V}} \amatrix{X}_{u},
\end{equation}
where $\beta$, $\gamma \in \reals^{p}$, and $\amatrix{X}_{u}$ is the Pauli-X operator acting on qubit $u$.
Therefore, the \gls{QAOA}['s] objective is
\begin{equation}
  \max_{\beta, \gamma \in \reals^{p}}
    \braket[\big]{+}{\amatrix{U}_{\mathrm{QAOA}}(\beta, \gamma)^{\A}
      \amatrix{H}_{\mathrm{MC}}
      \amatrix{U}_{\mathrm{QAOA}}(\beta, \gamma)}{+},
\end{equation}
with $\ket{+} \eqdef \amatrix{W}_{n} \ket{0}^{\otimes n}$, where $\amatrix{W}_{n}$ is the \param{n}-qubit Walsh--Hadamard transform.

The above optimization problem can be expressed as \pgls{PQCO} by introducing the vector of independent parameters
\begin{equation}
  \theta
    = (\beta_{1}, \dots, \beta_{p}, \gamma_{1}, \dots, \gamma_{p})
      \in \reals^{2p},
\end{equation}
so that $M = 2p$, taking
\begin{equation}
  \amatrix{O}
    = -\amatrix{H}_{\mathrm{MC}},
\end{equation}
and writing the ansatz as
\begin{equation}
  \amatrix{U}(\theta)
    = \xpar[\bigg]{\prod_{\ell = 1}^{p}
      \exp\xpar[\big]{-\im \beta_{\ell}\, \amatrix{B}}
      \exp\xpar[\big]{-\im \gamma_{\ell}\, \amatrix{H}_{\mathrm{MC}}}}
      \amatrix{W}_{n}.
\end{equation}

Observe that the generators $\amatrix{B}$ and $\amatrix{H}_{\mathrm{MC}}$ have integer spectra, and therefore have integer differences.
Indeed, as $\amatrix{H}_{\mathrm{MC}}$ is the Hamiltonian corresponding to the \textsc{Max-Cut} problem, its eigenvalues are cut values of the graph $\agraph{G}$, which are a subset of $\set[\big]{0, 1, \dots, \abs{\aset{E}}}$, and therefore are integers.
On the other hand, Pauli-X operators are simultaneously diagonalizable (as they commute), and have eigenvalues $\set{-1, 1}$.
Therefore, the eigenvalues of $\amatrix{B} = \sum_{v \in \aset{V}} \amatrix{X}_{v}$ are $\set[\big]{\sum_{u \in \aset{V}} \lambda_{u} \mid \lambda_{u} \in \set{-1, 1}} = \set{-n, -n + 2, \dots, n - 2, n}$, which are also integers.
In addition, by the above, $\amatrix{H}_{\mathrm{MC}}$ has operator norm at most $\abs{\aset{E}} \leq \abs{\aset{V}}^{2} = n^{2}$, and therefore spectral diameter at most $2n^{2}$, while $\amatrix{B}$ has spectral diameter $2n$, which are polynomial in $n$.
Finally, as the number of independent parameters $M = 2p$ is constant, and each parameter appears once in the circuit (thus, $K = M$), it follows that the fixed-level \gls{QAOA} ansatz satisfies \cref{ass:IntegerSpectralDifferences,ass:PQCComplexity,ass:ObservableComplexity}.

\subsubsection{\Glsfmtlongpl{HEA} with non-involutory generators}
\label{sec:HEAWithNonInvolutoryGenerators}

Although \pgls{PQC} may be expressed as multi/two-qubit gate decompositions in the Pauli basis~\citep{whitfield_2011_simulation}, quantum hardware does not necessarily \emph{implement} this decomposition directly.
This is because different quantum hardware allow for different choice of gates, qubit connectivity, an overall circuit architecture.
An example of such hardware is \emph{the trapped-ion quantum computer}~\citep{pino_2021_demonstration}, which natively implements all-to-all/long-range spin-spin coupling along with single-qubit rotations, rather than a compiled sequence of two-qubit Pauli gates, whose realization in such a device can be costly in both \emph{execution time} and \emph{accumulated error}~\citep{zhuang_2024_hardware-efficient}.
This motivates the design and optimization of \glspl{PQC} that natively support such hardware and are known as \glspl{HEA}.

As an example, we consider \pgls{HEA} for a trapped-ion quantum computer with $n$ qubits developed by \citet{zhuang_2024_hardware-efficient} with a constant number $M$ of independent parameters, given by
\begin{equation}
  \amatrix{U}_{\amatrix{J}}(t, \varphi)
    = \prod_{\ell = 1}^{D}
      \exp\xpar{-\im t_{\ell}\, \amatrix{H}_{\amatrix{J}}}
      \exp\xpar[\Big]{-\im \varphi_{\ell}\, \sum_{j = 1}^{n} \amatrix{Y}_{j}},
  \qquad
  \amatrix{H}_{\amatrix{J}}
    \eqdef \sum_{j < k} \amatrix{J}_{j, k}(\amatrix{X}_{j} \amatrix{X}_{k} + \amatrix{Y}_{j} \amatrix{Y}_{k}),
\end{equation}
where $D$ is polynomial in $n$ (polynomial circuit size), $\amatrix{Y}_{j}$ is the Pauli-Y operator acting on qubit $j$, and each entry $\amatrix{J}_{j, k}$ of the symmetric matrix $\amatrix{J} \in \reals^{n \times n}$ expresses the coupling strengths between qubits $j$ and $k$.
Here, the coupling matrix $\amatrix{J}$ follows \emph{a power-law decay} with respect to the distance between qubits~\citep{zhuang_2024_hardware-efficient}, i.e., the entries of $\amatrix{J}$ are given by
\begin{equation}
  \amatrix{J}_{j, k}
    \eqdef \begin{cases}
      \frac{J_{0}}{\abs{j - k}^{\alpha}}
        & \text{if $j \neq k$}; \\
      0
        & \text{otherwise},
    \end{cases}
\end{equation}
for some constant $J_{0} > 0$ and $\alpha \geq 0$.
Then, given a Hermitian observable $\amatrix{O}$ that satisfies \cref{ass:ObservableComplexity}, the \gls{PQC} optimization problem for this ansatz boils down to finding the minimum value
\begin{equation}
  \min_{t, \varphi \in \reals^{D}} f_{\amatrix{J}}(t, \varphi),
  \ \text{where} \ 
  f_{\amatrix{J}}(t, \varphi)
    = \braket[\big]{+}{\amatrix{U}_{\amatrix{J}}(t, \varphi)^{\A} \amatrix{O} \amatrix{U}_{\amatrix{J}}(t, \varphi)}{+},
  \qquad
  \ket{+} 
    \equiv \ket{+}^{\otimes n}.
\end{equation}
To ease the exposition, we further assume that the eigenvalues of $\amatrix{J}$ are integer multiples of $1 / L$ for some quantization level $L \in \naturals$.
This assumption is not essential.
Indeed, it can be relaxed via a polynomial-time rescaling of parameter units.
Moreover, if $\norm{\amatrix{O}}$ is polynomially bounded in $n$ (cf.~\cref{ass:ObservableComplexity}), one can reduce in polynomial time to this quantized setting. 

If the number of independent parameters $M$ of the \gls{HEA} is constant, then the above \gls{PQC} satisfies \cref{ass:IntegerSpectralDifferences,ass:PQCComplexity}.
To see this, fix $L$, and define $\theta \in \reals^{M}$ to be the $M$ independent parameters of $\amatrix{U}_{\amatrix{J}}$.
Then the \gls{PQC} can be expressed in terms of the independent parameters $\theta$ as
\begin{equation}
  \amatrix{U}_{\amatrix{J}}(t, \varphi)
    = \prod_{\ell = 1}^{D}
      \exp\xpar{-\im \theta_{u_{\ell}} L\, \amatrix{H}_{\amatrix{J}}}
      \exp\xpar[\Big]{-\im \theta_{v_{\ell}}\, \sum_{j = 1}^{n} \amatrix{Y}_{j}},
\end{equation}
where $u_{1}$, \dots, $u_{D}$, $v_{1}$, \dots, $v_{D} \in \range{M}$ are 
such that $t_{\ell} / L = \theta_{u_{\ell}}$, and $\varphi_{\ell} = \theta_{v_{\ell}}$ for all $\ell \in \range{D}$.
Now, observe that the spectrum of $L\, \amatrix{H}_{\amatrix{J}}$ has integer differences, and its diameter is bounded by a polynomial in $n$.

Indeed, by the standard Jordan--Wigner/fermionic second-quantization formalism (see, e.g., \citet{lieb_1961_soluble} and standard treatments of fermionic second quantization), the eigenvalues of $L\, \amatrix{H}_{\amatrix{J}}$ are
\begin{equation}
  \aset{L}_{L\, \amatrix{H}_{\amatrix{J}}} 
    \eqdef \set[\bigg]{2 L \sum_{j \in \aset{S}} \lambda_{j} \given \aset{S} \subseteq \range{n}},
\end{equation}
where $\lambda_{1}$, \dots $\lambda_{n} \in \reals$ are the eigenvalues of $\amatrix{J}$.
Thus, since by assumption $L \lambda_{j} \in \integers$ for all $j$, the eigenvalues of $L\, \amatrix{H}_{\amatrix{J}}$ are integers, and therefore have integer differences.
Furthermore, the spectral diameter of $L\, \amatrix{H}_{\amatrix{J}}$ is at most
\begin{subequations}
\begin{align}
  \Delta_{L\, \amatrix{H}_{\amatrix{J}}}
    &\eqdef \max \aset{L}_{L\, \amatrix{H}_{\amatrix{J}}} - \min \aset{L}_{L\, \amatrix{H}_{\amatrix{J}}} \\
    &= 2 \sum_{j = 1}^{n} \abs{L \lambda_{j}} \\
    &\leq 2 L n \max_{j \in \range{n}} \abs{\lambda_{j}} \\
    &= 2 L n \norm{\amatrix{J}}_{2} \\
    &\leq 2 L n \norm{\amatrix{J}}_{\infty} \\
    &= 2 L n \max_{j \in \range{n}} \sum_{k = 1}^{n} \frac{J_{0}}{\abs{j - k}^{\alpha}} \\
    &\leq 2 L J_{0} n^{2},
\end{align}
\end{subequations}
which is a polynomial in $n$ for every fixed $L$.

The remaining arguments are straightforward.
In particular, since $\amatrix{Y}_{j}$ is involutory for all $j \in \range{n}$, the spectrum of $\sum_{j = 1}^{n} \amatrix{Y}_{j}$ has integer differences and diameter at most $2n$, which are at most polynomial in $n$.
Thus, the \gls{HEA} satisfies \cref{ass:IntegerSpectralDifferences,ass:PQCComplexity}.

\section{Preliminaries}
\label{sec:Preliminaries}

To make this exposition self-contained, we now introduce the necessary background on Hermitian polynomial optimization (cf.~\cref{sec:HermitianPolynomialOptimization}) and \gls{HTP} optimization (cf.~\cref{sec:HTPOptimization}), which are the main mathematical tools used in this work, as well as the multidimensional \gls{FFT} (cf.~\cref{sec:TrigonometricPolynomialsFFT}), which is a key algorithmic ingredient in our proposed method for solving \gls{PQC} optimization problems.

\subsection{Hermitian polynomial optimization--the complex moment/\glsfmtshort{SOS} hierarchy} 
\label{sec:HermitianPolynomialOptimization}
 
\Gls{HTP} optimization can be placed within the more general class of Hermitian polynomial optimization, for which the complex moment/\glsdashedorshort{SOS} hierarchy has already been well-developed~\citep{laurent_2009_sums-of-squares}. 
Therefore, to keep the discussion self-contained, we devote this \lcnamecref{sec:HermitianPolynomialOptimization} to reviewing the topic of Hermitian polynomial optimization and the complex moment/\glsdashedorshort{SOS} hierarchies.

A Hermitian polynomial~\citep{dangelo_2009_polynomial} is a function of the form
\begin{equation}
\label{eq:HermitianPolynomial}
  f(z, \conj{z})
    = \sum_{\alpha, \beta} f_{\alpha, \beta} z^{\alpha} \conj{z}^{\beta},
    \qquad z \in \complex^{M},
\end{equation}
where $\alpha, \beta \in \naturals^{M}$, $f_{\alpha, \beta} = \conj{f_{\beta, \alpha}} \in \complex$ for all $\alpha$, $\beta$, and only finitely many $f_{\alpha, \beta}$ are nonzero.
The condition $f_{\alpha, \beta} = \conj{f_{\beta, \alpha}}$ ensures that $f(z, \conj{z})$ is real-valued for all $z \in \complex^{M}$.
Subsequently, a Hermitian polynomial optimization problem~\citep{josz_2018_lasserre,wang_2022_exploiting,jiang_2014_alternating} is given by
\begin{equation}
\label{eq:HermitianPolynomialOptimization}
  f^{\star}
    \eqdef \inf_{z \in \aset{S}} f(z, \conj{z}),
  \ \text{and} \
  \aset{S}
    \eqdef \set[\Big]{z \in \complex^{M} \given g_{i}(z, \conj{z}) \geq 0 \,\forall i \in \range{m_{g}}, \ h_{i}(z, \conj{z}) = 0 \,\forall i \in \range{m_{h}}},
\end{equation}
with $f$, $g_{i}$, $h_{i}$ Hermitian polynomials.
Note that, since Hermitian polynomial optimization includes real polynomial optimization as a special case, it is, in general, \pgls{NP}\nobreakdash-hard problem~\citep{murty_1987_np-complete}.

As all involved polynomials in \eqref{eq:HermitianPolynomialOptimization} are Hermitian, and therefore real-valued, one may set $z_{j} = x_{j} + \im y_{j}$ and $\conj{z}_{j} = x_{j} - \im y_{j}$, for all $j \in \range{M}$, obtaining an equivalent real polynomial optimization in $(x, y) \in \reals^{2M}$, which can be handled by the standard real Lasserre moment/\glsdashedorshort{SOS} hierarchy~\citep{lasserre_2006_sum-of-squares,parrilo_2003_semidefinite}.  
An alternative approach, proposed by \citet{josz_2018_lasserre} for large-scale optimal power flow problems, is to work directly over complex variables. 
This relies on the notion of a Hermitian \glsdashedorshort{SOS} polynomial, namely a Hermitian polynomial $\asos(z, \conj{z})$ of the form
\begin{equation}
\label{eq:HermitianSOS}
  \asos(z, \conj{z})
    = \sum_{i} \abs[\big]{p_{i}(z)}^{2},
  \ \text{where} \
  p_{i}(z)
   = \sum_{\alpha} p_{i, \alpha} z^{\alpha},
\end{equation}
where each $p_{i}$ is \emph{holomorphic} (depends only on $z$). 
If one were to allow the $p_{i}$ to depend on both $z$ and $\conj{z}$, the resulting hierarchy is equivalent to the real Lasserre \glsdashedorshort{SOS} hierarchy obtained by the aforementioned realification procedure.

The starting point of the \glsdashedorshort{SOS} hierarchy is to reformulate the Hermitian polynomial optimization problem in \eqref{eq:HermitianPolynomialOptimization} as the problem of finding the largest $\lambda$ for which $f(z, \conj{z}) - \lambda$ is nonnegative over $\aset{S}$. 
Then the \glsdashedorshort{SOS} hierarchy computes \emph{lower bounds} on $f^{\star}$ by searching for algebraic decompositions that provide sufficient conditions for the nonnegativity of $f(z, \conj{z}) - \lambda$ over $\aset{S}$.
Specifically, in this work, we focus on the so-called Putinar-type decompositions expressed as
\begin{equation}
\label{eq:HermitianSOSCertificate}
  f(z, \conj{z}) - \lambda
    = \asos_{0}(z, \conj{z}) 
      + \sum_{i = 1}^{m_{g}} \asos_{i}(z, \conj{z}) \,g_{i}(z, \conj{z})
      + \sum_{i = 1}^{m_{h}} q_{i}(z, \conj{z}) \,h_{i}(z, \conj{z}),
\end{equation}
where each $\asos_{i}$ is a Hermitian \glsdashedorshort{SOS} polynomial, and each $q_{i}$ is an arbitrary Hermitian polynomial. 
Finally, optimizing over $\lambda$ and fixed-degree certificates $\asos_{i}$, $q_{i}$
reduces to solving a semidefinite program.

\paragraph{The dual formulation.} On the dual side, the starting point is that $f^{\star}$ can equivalently be obtained by taking the infimum of the integral $\int_{\aset{S}} f \,d \mu$ over probability measures $\mu$ supported on $\aset{S}$. 
By defining complex moments as
\begin{equation}
  y_{\alpha, \beta} 
    \eqdef \int_{\aset{S}} z^{\alpha} \conj{z}^{\beta} \,d\mu,
\end{equation}
for all $\alpha$, $\beta \in \naturals^{M}$, we can see that the objective function $f(z, \conj{z})$ may be reformulated as $\sum_{\alpha, \beta} f_{\alpha, \beta} \,y_{\alpha, \beta}$, which is linear in $y$. 

Similarly, the inequality and equality constraints on $\aset{S}$ translate into nonlinear and linear conditions on $y$, respectively.
To see this, let $v(z)$ denote the (infinite) vector of all monomials in $z = (z_{1}, \dots, z_{M})$.
We write $v_{d}(z)$ for the sub-vector of monomials of degree (total or max depending on the context) at most $d$.
Then, any holomorphic polynomial of degree at most $d$ can be written as a linear combination of the basis elements in $v_{d}(z)$, i.e., $p(z) = c^{\A} v_{d}(z)$, where $c$ is the vector of complex coefficients (of size equal to the one of $v_{d}(z)$). 
As $\abs[\big]{p(z)}^{2}  = c^{\A} \,\xpar[\big]{v_{d}(z) v_{d}(z)^{\A}} \,c$, it follows that
\begin{equation}
\label{eq:TruncatedMomentMatrixPSD}
  \int_{\aset{S}} \abs[\big]{p(z)}^{2} \,d \mu 
    = c^{\A} \,\xpar[\bigg]{\int_{\aset{S}} v_{d}(z) v_{d}(z)^{\A} \,d \mu} \,c
    \geq 0.
\end{equation}
Similarly, since the Hermitian polynomials $g_{i}$ are nonnegative over $\aset{S}$ (by the definition of $\aset{S}$), we also have that
\begin{equation}
\label{eq:TruncatedLocalizingMomentMatrixPSD}
  \int_{\aset{S}} \abs[\big]{p(z)}^{2} \,g_{i}(z, \conj{z}) \,d \mu 
    = c^{\A} \,\xpar[\bigg]{\int_{\aset{S}} v_{d}(z) v_{d}(z)^{\A} \,g_{i}(z, \conj{z}) \,d \mu} \,c
    \geq 0,
    \qquad \forall i \in \range{m_{g}}.
\end{equation}
This motivates the definition of the infinite-dimensional matrices 
\begin{equation}
\label{eq:MomentMatrices}
  \amatrix{M}(y)
    = \int_{\aset{S}} v(z) v(z)^{\A} \,d \mu 
  \quad \text{and} \quad 
  \amatrix{M}(g_{i} \,y)
    = \int_{\aset{S}} v(z) v(z)^{\A} g_{i}(z, \conj{z}) \,d \mu,
    \qquad \forall i \in \range{m_{g}},
\end{equation}
referred to respectively as \emph{the moment matrix} and \emph{the localizing moment matrix}.  Crucially, by \cref{eq:TruncatedMomentMatrixPSD,eq:TruncatedLocalizingMomentMatrixPSD}, any finite truncation of these matrices is positive semidefinite.
Localizing matrices $\amatrix{M}(h_{i} \,y)$ corresponding to the equality constraints $h_{i}(z, \conj{z}) = 0$ can be defined analogously and it can be shown that they satisfy $\amatrix{M}(h_{i} \,y) = 0$ for each $i \in \range{m_{h}}$.

For each $g_{i}$ (analogously for $h_{i}$), the \param{(\alpha, \beta)}-entry of a localizing moment matrix is given by
\begin{equation}
\label{eq:LocalizingMomentMatrixEntries}
  \amatrix{M}(g_{i} \,y)_{\alpha, \beta} 
    = \int_{\aset{S}} z^{\alpha} \conj{z}^{\beta} \,g_{i}(z, \conj{z}) \,d \mu 
    = \sum_{\gamma, \delta} (g_{i})_{\gamma, \delta} \,y_{\alpha + \gamma, \;\beta + \delta}.
\end{equation}
Therefore, for each truncation $\amatrix{M}_{d}(g_{i} \,y)$ of the localizing moment matrix $\amatrix{M}(g_{i} \,y)$, the \glsdashedorshort{PSD} constraint $\amatrix{M}_{d}(g_{i} \,y) \succeq 0$ is a linear matrix inequality in the vector of variables $y$, whereas the constraints $\amatrix{M}_{d}(h_{i} \,y) = 0$, for each truncation of the localizing moment matrix $\amatrix{M}(h_{i} \,y)$, yield linear equality constraints in terms of $y$. 
Then, by considering a sequence of finite (and increasing) truncations $\amatrix{M}_{d}(y)$ of the moment matrix and $\amatrix{M}_{d}(g_{i} \,y)$ and $\amatrix{M}_{d}(h_{i} \,y)$ for the localizing moment matrices, for $d \geq 0$, e.g., by restricting to monomials of degree at most $d$, we obtain a hierarchy of nondecreasing lower bounds on the optimal value $f^{\star}$, each efficiently computable via semidefinite programming. 

As a concrete example, \citet{josz_2018_lasserre} consider truncations based on the total degree.
At relaxation order $d$, one optimizes over truncated \emph{pseudo-moment sequences} $(y_{\alpha,\beta})$ with $\norm{\alpha}_{1}$, $\norm{\beta}_{1} \leq d$, and imposes the \glsdashedorshort{PSD} constraints $\amatrix{M}_{d}(y) \succeq 0$ and $\amatrix{M}_{d - k_{i}}(g_{i} \,y) \succeq 0$, where $k_{i}$ is the maximum of $\norm{\gamma}_{1}$ and $\norm{\delta}_{1}$ across all exponents $\gamma$, $\delta$ such that $(g_{i})_{\gamma, \delta} \neq 0$. Similarly, it imposes the linear constraints $\amatrix{M}_{d - k'_{i}}(h_{i} \,y) = 0$ with $k'_{i}$ defined analogously to $k_{i}$ for each $h_{i}$.
To understand why the localizing moment matrix $\amatrix{M}(g_{i} \,y)$ (resp. $\amatrix{M}(h_{i} \,y)$) is truncated at order $d - k_{i}$ (resp. $d - k'_{i}$), note that the \param{(\alpha, \beta)}-entry of $\amatrix{M}(g_{i} \,y)$ involves the moment $y_{\alpha + \gamma, \;\beta + \delta}$. 
Therefore, to ensure that $\norm{\alpha + \gamma}_{1} \leq d$ and $\norm{\beta + \delta}_{1} \leq d$, so that these variables are available, it is sufficient to require $\norm{\alpha}_{1}$, $\norm{\beta}_{1} \leq d - k_{i}$, and thus, the localizing moment matrix is consistently indexed up to degree $d - k_{i}$ (resp. $d - k'_{i}$).
Summarizing, the order\nobreakdash-$d$ relaxation of the moment hierarchy is given by 
\begin{equation}
\label{eq:HermitianPolynomialOptimizationMomentRelaxation}
\setobjectivesense{\inf}[y]%
\begin{aligned}
  f_{d}^{\star}
    &\eqdef \begin{aligned}[t]
      &\objectivesense%
        && \sum_{\alpha, \beta} f_{\alpha, \beta} \,y_{\alpha, \beta} \\
      &\subjectto%
        &&\begin{aligned}[t]
          &y_{0, 0} = 1; \\
          &\amatrix{M}_{d}(y) \succeq 0; \\
          &\amatrix{M}_{d - k_{i}}(g_{i} \,y) \succeq 0,
            \ \forall i \in \range{m_{g}}; \\
          &\amatrix{M}_{d - k'_{i}}(h_{i} \,y) = 0,
            \ \forall i \in \range{m_{h}},
      \end{aligned}
    \end{aligned}
\end{aligned}
\end{equation}
where the decision variable is the vector $y = (y_{\alpha, \beta})$ with $\norm{\alpha}_{1}$, $\norm{\beta}_{1} \leq d$.

\Citet{josz_2018_lasserre} also show that if one of the constraints is \emph{a sphere constraint} of the form 
\begin{equation}
  \sum_{i = 1}^{M} \abs{z_{i}}^{2} 
    = R^{2}
\end{equation}
for some $R > 0$, then the sequence of relaxations $(f_{d}^{\star})_{d}$ exhibits asymptotic convergence, i.e., the relaxation values $f_{d}^{\star}$ converge to the global optimum $f^{\star}$ as $d \to \infty$. 
In addition, under appropriate rank conditions, \emph{finite convergence} can be achieved, and one can extract global minimizers from the moment matrix~\citep[Proposition~4.1]{josz_2018_lasserre}.
Importantly, \citep[Proposition~4.1]{josz_2018_lasserre} does not cover truncations based on the max degree.
We prove the analogous result for max-degree truncations in \cref{sec:PQCOptimizerExtraction}.

\subsection{\Glsfmtlong{HTP} optimization}
\label{sec:HTPOptimization}

An important special case of the Hermitian polynomial optimization problem in \eqref{eq:HermitianPolynomial} considers the optimization of a Hermitian polynomial over the (complex) \param{M}-torus
\begin{equation}
  \complextorus^{M}
    \eqdef \set[\Big]{z \in \complex^{M} \given \abs{z_{i}} = 1 \,\forall i \in \range{M}}.
\end{equation}
In this case, there are no inequality constraints, and the only constraints are $h_{i}(z, \conj{z}) = z_{i} \conj{z}_{i} - 1 = 0$, for all $i \in \range{M}$.
In addition, the moment hierarchy at order\nobreakdash-$d$ corresponds to the semidefinite program
\begin{equation}
\label{eq:HTPOptimizationMomentRelaxation}
\setobjectivesense{\inf}[y]%
\begin{aligned}
  &\objectivesense%
    &&\sum_{\alpha, \beta} f_{\alpha, \beta}\, y_{\alpha, \beta} \\
  &\subjectto%
    &&\begin{aligned}[t]
      &y_{0, 0} = 1; \\
      &\amatrix{M}_{d}(y) \succeq 0 \ \text{and Toeplitz}.
  \end{aligned}
\end{aligned}
\end{equation}
Here the Toeplitz structure comes from the equality constraints $z_{i} \conj{z}_{i} - 1 = 0$. Indeed, this polynomial has only two nonzero monomials, namely $(e_{i}, e_{i})$ and $(0, 0)$, where $e_{i} \in \naturals^{M}$ is \emph{a vector of exponents} with a single nonzero entry equal to $1$ at $i$, and therefore $z^{e_{i}} = z_{i}$.
Thus, from the definition of the moment matrix in \eqref{eq:MomentMatrices}, for any $\alpha$, $\beta \in \naturals^{M}$ we obtain
\begin{equation}
  y_{\alpha + e_{i}, \;\beta + e_{i}} - y_{\alpha, \beta} 
    = 0, 
    \qquad \text{$\forall \alpha$, $\beta$ with $\norm{\alpha}_{\infty}$, $\norm{\beta}_{\infty} \leq d$}.
\end{equation}
This implies that the truncated moment matrix $\amatrix{M}_{d}(y)$ is a generalized Toeplitz matrix, i.e., its entries depend only on index differences.

On the primal side, the order\nobreakdash-$d$ \glsdashedorshort{SOS} relaxation minimizes the polynomial $f(z, \conj{z})$ by finding the maximal $\lambda \in \reals$ for which $f(z, \conj{z}) - \lambda$ is a Hermitian \glsdashedorshort{SOS} polynomial of total degree $d$ on $z$ and $\conj{z}$; that is
\begin{equation}
\label{eq:HTPOptimizationSOSRelaxation}
\setobjectivesense{\sup}[\lambda, \asos, q_{1}, \dots, q_{M}]%
\begin{aligned}
  &\objectivesense%
    &&\lambda \\
  &\subjectto%
    &&\begin{aligned}[t]
      &f(z, \conj{z}) - \lambda = \asos(z, \conj{z}) + \sum_{i = 1}^{M} q_{i}(z, \conj{z}) \,(z_{i} \conj{z}_{i} - 1), 
  \end{aligned}
\end{aligned}
\end{equation}
where $\asos$ is a Hermitian \glsdashedorshort{SOS} polynomial of max degree $d$ and $q_{1}$, \dots, $q_{M}$ are Hermitian polynomials of max degree $d - 1$.

\paragraph{\Glsfmtlongpl{HTP}.} \Pgls{HTP} in $M$ variables is a function of the form
\begin{equation}
\label{eq:HermitianTrigonometricPolynomial}
  f(z)
    = \sum_{\alpha \in \integers^{M}} f_{\alpha}\, z^{\alpha}, 
    \qquad z \in \complextorus^{M},
\end{equation}
where $\alpha \in \integers^{M}$, $f_{-\alpha} = \conj{f_{\alpha}} \in \complex$ for all $\alpha$, and only finitely many $f_{\alpha}$ are nonzero.
Once again, the condition $f_{-\alpha} = \conj{f_{\alpha}}$ ensures that $f(z)$ is real-valued for all $z \in \complextorus^{M}$.
We say that $f$ has \emph{max degree} $d$ if every nonzero term satisfies $\abs{\alpha_{j}} \leq d$ for all $j \in \range{M}$, or equivalently $\norm{\alpha}_{\infty} \leq d$.
Equivalently, writing each $z_{j} = \exp(\im \theta_{j})$ with $\theta \in \ccinterval{-\pi}{\pi}^{M}$, we obtain the \emph{angle representation}
\begin{equation}
  f(\theta)
    = \sum_{\alpha \in \integers^{M}} f_{\alpha}\, \exp\xpar[\big]{\im \inner{\alpha}{\theta}},
    \qquad \theta \in \ccinterval{-\pi}{\pi}^{M},
\end{equation}
In this form, the Hermitian symmetry $f_{-\alpha} = \conj{f_{\alpha}}$ again guarantees that $f(\theta)$ is real-valued for all $\theta \in \ccinterval{-\pi}{\pi}^{M}$.

It is not difficult to see that the \emph{\glsxtrlong{HTP} optimization problem} 
\begin{equation}
  \min_{\theta \in \ccinterval{-\pi}{\pi}^{M}} f(\theta),
\end{equation} 
with $f$ \pgls{HTP}, is a Hermitian polynomial optimization problem over the \param{M}-torus. 
In particular, we have the following \lcnamecref{thm:HTPOptimizationHPOptimizationEquivalence}. 
\begin{lemma}
\label{thm:HTPOptimizationHPOptimizationEquivalence}
  Any Hermitian polynomial optimization problem over the \param{M}-torus $\complextorus^{M}$ is equivalent, via the substitution $z_{j} = \exp(\im \theta_{j})$, to \pgls{HTP} optimization problem. 
\end{lemma}
\begin{proof}
  Consider the Hermitian complex polynomial
  \begin{equation}
    g(z, \conj{z})
      = \sum_{\alpha, \beta} g_{\alpha, \beta} z^{\alpha} \conj{z}^{\beta},
      \qquad z \in \complex^{M}
  \end{equation}
  with max degree $d$ meaning that for all non-zero terms we have $\max \set{\norm{\alpha}_{\infty}, \norm{\beta}_{\infty}} \leq d$.
  Then, for $z \in \complextorus^{M}$, substituting the parametrization $z_{j} = \exp(\im \theta_{j})$ with $\theta_{j} \in \ccinterval{-\pi}{\pi}$, for all $j \in \range{M}$, we obtain
  \begin{equation}
    g(z, \conj{z})
      = g\xpar[\big]{\exp(\im \theta_{j}), \exp(-\im \theta_{j})}
      = \sum_{\alpha, \beta} g_{\alpha, \beta}\, \exp\xpar[\big]{\im \inner{\alpha - \beta}{\theta}}.
  \end{equation}
  Moreover, by setting $\gamma = \alpha - \beta$, we can rewrite the expression as the Fourier series
  \begin{equation}
    f(\theta) 
      = \sum_{\substack{\gamma \in \integers^{M} \\ \norm{\gamma}_{\infty} \leq d}} f_{\gamma}\, \exp\xpar[\big]{\im \inner{\gamma}{\theta}},
    \ \text{where} \
    f_{\gamma} 
      \eqdef \sum_{\alpha - \beta = \gamma} g_{\alpha, \beta}.
  \end{equation}
  Finally, since $g$ is Hermitian, we have $\conj{g_{\alpha, \beta}} = g_{\beta, \alpha}$ which in turn implies that $f_{-\gamma} = \conj{f_{\gamma}}$; thus, $f$ is \pgls{HTP}, and therefore, the optimization of $g$ over $\complextorus^{M}$ is equivalent to the optimization of $f$ over $\ccinterval{-\pi}{\pi}^{M}$.
\end{proof} 

\paragraph{Max-degree trigonometric \glsfmtshort{SOS} hierarchy.} In view of \cref{thm:HTPOptimizationHPOptimizationEquivalence}, since any \gls{HTP} optimization problem is equivalent to a Hermitian polynomial optimization problem over the \param{M}-torus, we immediately obtain a max-degree moment/\glsdashedorshort{SOS} hierarchy for such problems, as described above.
We refer to the corresponding hierarchies as the \emph{max-degree trigonometric moment/\glsdashedorshort{SOS} hierarchies}.
Notably, the trigonometric \glsdashedorshort{SOS} hierarchy was already developed earlier in the context of signal processing applications; see \cite{dumitrescu_2007_positive}. 
Two main differences between that earlier work and the approach of \cite{josz_2018_lasserre} are that 
\begin{enumerate*}[label=(\roman*)]
  \item the former focuses solely on the \glsdashedorshort{SOS} side (no dual formulation), and 
  \item it does not provide conditions under which an optimizer can be extracted from the hierarchy.
\end{enumerate*}

More recently, \citet{bach_2023_exponential} provide explicit error bounds for the max-degree trigonometric \gls{SOS} hierarchy, quantifying how fast $f_{d}^{\star} \to f^{\star}$ as $d \to \infty$; see also \citet{laurent_2026_overview} for an overview of convergence rates for \glsdashedorshort{SOS} hierarchies.
Relative to the total-degree truncation described in \cref{sec:HermitianPolynomialOptimization}, their analysis is stated for a \emph{max-degree} (i.e., $\norm{\cdot}_{\infty}$) truncation, which is also the natural choice for \glspl{HTP} supported on frequencies $\set{\alpha \in \integers^{M} \given \norm{\alpha}_{\infty} \leq d}$.
Concretely, letting $\amatrix{M}_{d}(y) \in \complex^{N \times N}$ be the truncated moment matrix with $N = (2d + 1)^{M}$ with entries
\begin{equation}
  {\amatrix{M}_{d}(y)}_{\alpha, \beta} 
    \eqdef y_{\alpha, \;\beta}
  \qquad \alpha, \beta \in \integers^{M} \text{with $\norm{\alpha}_{\infty}$, $\norm{\beta}_{\infty} \leq d$},
\end{equation}
the order-$d$ max-degree trigonometric moment relaxation is once again given by the semidefinite program
\begin{equation}
\label{eq:TrigonometricMomentRelaxation}
  \setobjectivesense{\inf}[y]%
  \begin{aligned}
    f_{d}^{\star}
      &= \begin{aligned}[t]
        &\objectivesense%
          &&\sum_{\alpha, \beta} f_{\alpha, \beta}\, y_{\alpha, \beta} \\
        &\subjectto%
          &&\begin{aligned}[t]
            &y_{0, 0} = 1; \\
            &\amatrix{M}_{d}(y) \succeq 0 \ \text{and Toeplitz}.
        \end{aligned}
      \end{aligned}
  \end{aligned}
\end{equation}
Equivalently on the \gls{SOS} side, this corresponds to restricting the Putinar-type certificate to max-degree $d$.

Specifically, \citet[Theorem~1]{bach_2023_exponential} show that truncation at level $d$ yields a quadratic decay of the approximation error, namely
\begin{equation}
 \abs{f_{d}^{\star} - f^{\star}} 
    = \bigoh(1 / d^{2}),
\end{equation}
without additional assumptions. 
Moreover, if $f$ is $\aset{C}^{\infty}$ and its derivatives satisfy the growth condition in \citet[Theorem~2]{bach_2023_exponential}, then the hierarchy converges at an exponential rate.
In \cref{sec:GlobalPQCOptimization}, we use the max-degree trigonometric \glsdashedorshort{SOS} hierarchy to obtain the exponential convergence guarantees of \citet{bach_2023_exponential} for \gls{PQC} optimization.

\subsection{Trigonometric polynomials and the \glsfmtlong{FFT}}
\label{sec:TrigonometricPolynomialsFFT}

The \gls{FFT}~\citep{cooley_1965_algorithm} can be used as an efficient method for converting between the two canonical representations of a trigonometric polynomial $f$; these are,
\begin{enumerate*}[label=(\roman*)]
  \item the vector of coefficients $(f_{\alpha})_{\alpha}$, and 
  \item the vector of evaluations of $f$ on a \emph{uniform sampling grid}.
\end{enumerate*}
Crucially, this conversion can be performed in \emph{quasi-linear time} $\bigoh(N \log N)$, where $N$ is the number of sampled points, i.e., the size of either of the representation vectors (which are of the same size).
For completeness, we outline this procedure in order to highlight two important aspects that are particular to our setting: 
\begin{enumerate}
  \item Its application to trigonometric polynomials, where the natural sampling grid is on the unit circle. 
  \item Its extension to the multivariate case, where \gls{FFT} subroutines are performed in parallel along each coordinate.
\end{enumerate}

\paragraph{Univariate case.} We begin with the univariate case in order to familiarize ourselves with the required notation.
Consider the univariate complex trigonometric polynomial (not necessarily Hermitian)
\begin{equation}
  f(z) 
    = \sum_{n = -d}^{d} f_{n}\, z^{n},
    \qquad z \in \complex,
\end{equation}
and sample $f$ at $N = 2d + 1$ uniformly spaced points $z_{-d}$, \dots, $z_{d}$ on the unit circle, i.e.,
\begin{equation}
  z_{t} 
    \eqdef \exp(\im \theta_{t}), 
  \qquad
  \theta_{t} 
    \eqdef \frac{2 \pi t}{N},
  \qquad t = -d, \dots, d.
\end{equation}

Let $\omega^{t} \eqdef \exp(-2 \pi \im t / N)$, for $t = -d, \dots, d$, be the \param{t}-th root of unity, so that $z_{t} = \omega^{-t}$.
We define the vector of evaluations 
\begin{equation}
  \amatrix{p}
    \eqdef \xpar[\big]{f(z_{-d}), \dots, f(z_{d})},
  \ \text{where} \
  f(z_{t}) 
    = \sum_{n = -d}^{d} f_{n}\, \omega^{-t n},
\end{equation}
and let $\amatrix{\Phi} \in \complex^{N \times N}$ be the matrix defined by 
\begin{equation}
  \Phi_{t, n} 
    \eqdef \omega^{-t n},
    \qquad t, n = -d, \dots, d.
\end{equation}
Then it is easy to see that the evaluation vector can be written as the matrix-vector product
\begin{equation}
  \amatrix{p}
    = \amatrix{\Phi}\, \amatrix{f},
\end{equation}
where $\amatrix{f} = (f_{-d}, \dots, f_{d})$ is the vector of coefficients of $f$.

Since $N = 2d + 1$, the index range $-d$, \ldots, $d$ corresponds to all distinct frequencies.  
Thus $\amatrix{\Phi}$ is a row and column reordering of the standard $N \times N$ \gls{DFT} matrix $\amatrix{F}_{N}$ given by $(\amatrix{F}_{N})_{j, k} \eqdef \omega^{j k}$ for $j$, $k = 0, \dots, N - 1$.
Hence, $\amatrix{\Phi}$ is unitary up to the scalar $N$, i.e.
\begin{equation}
  \amatrix{\Phi}^{\A} \amatrix{\Phi} 
    = N\, \identitymatrix_{N},
\end{equation}
This gives the inverse mapping
\begin{equation}
  \amatrix{f} 
    = \frac{1}{N}\, \amatrix{\Phi}^{\A} \amatrix{p},
  \ \text{and} \
  f_{n} 
    = \frac{1}{N} \sum_{t = -d}^{d} f(z_{t})\, \omega^{t n}.
\end{equation}
Moreover, using the \gls{FFT}~\citep{cooley_1965_algorithm}, the inverse \gls{DFT} $\amatrix{f} = \frac{1}{N}\, \amatrix{\Phi}^{\A} \amatrix{p}$ can be computed in quasi-linear time $\bigoh(N \log N)$ (versus $\bigoh(N^{2})$ for vanilla matrix inversion).  

\paragraph{Multivariate case.} Now, consider a trigonometric polynomial in $M$ variables
\begin{equation}
  f(z_{1}, \dots, z_{M}) 
    = \sum_{n_{1} = - d_{1}}^{d_{1}} \dots \sum_{n_{M} = - d_{M}}^{d_{M}} f_{n_{1}, \dots, n_{M}}\, \prod_{j = 1}^{M} z_{j}^{n_{j}},
    \qquad z = (z_{1}, \dots, z_{M}) \in \complex^{M},
\end{equation}
supported in the multi-index range $-d_{j}$, \dots, $d_{j}$ for each variable $z_{j}$, $j = 1, \dots, M$, and let $\amatrix{F} \in \complex^{N_{1} \times \cdots \times N_{M}}$ be the tensor of coefficients of $f$ indexed so that
\begin{equation}
  \amatrix{F}[n_{1}, \dots, n_{M}] 
    \eqdef f_{n_{1}, \dots, n_{M}},
    \qquad n_{j} = -d_{j}, \dots, d_{j}, \ j = 1, \dots, M.
\end{equation}

Define $N_{j} \eqdef 2d_{j} +1$ for each mode $j \in \range{M}$.
We sample $f$ on a tensor-product sampling grid of uniformly spaced points on the unit circle; that is, for each mode $j \in \range{M}$, we set
\begin{equation}
  z_{j, t_{j}} 
    \eqdef \exp\xpar[\big]{2 \pi \im t_{j} / N_{j}},
  \qquad t_{j} = -d_{j}, \dots, d_{j}.
\end{equation}
Thus, the sampling points in the $j$-th variable lie at $N_{j}$ equally spaced angles around the unit circle. 
The entire sampling grid is then the set of points
\begin{equation}
  (z_{1, t_{1}}, \dots, z_{M, t_{M}}),
  \qquad t_{j} = -d_{j}, \dots, d_{j}.
\end{equation}

For each $j$, set $\omega_{N_{j}}^{t} \eqdef \exp(- 2\pi \im t / N_{j})$, for $t = -d_{j}, \dots, d_{j}$, as the roots of unity of order $N_{j}$, so that $z_{j, t_{j}} = \omega_{N_{j}}^{-t_{j}}$, and let $\amatrix{P} \in \complex^{N_{1} \times \cdots \times N_{M}}$ be the tensor of evaluations of $f$ on the sampling grid, indexed so that
\begin{equation}
  \amatrix{P}[t_{1}, \dots, t_{M}] 
    \eqdef f(z_{1, t_{1}}, \dots, z_{M, t_{M}}) 
    = \sum_{n_{1} = -d_{1}}^{d_{1}} \dots \sum_{n_{M} = -d_{M}}^{d_{M}} f_{n_{1}, \dots, n_{M}}\, \prod_{j = 1}^{M} \omega_{N_{j}}^{-t_{j} n_{j}}.
\end{equation}

For each mode $j$, define the (reordered) square \gls{DFT} matrix $\amatrix{\Phi}^{(j)} \in \complex^{N_{j} \times N_{j}}$ by entries
\begin{equation}
  \amatrix{\Phi}^{(j)}_{t_{j}, n_{j}}
     \eqdef \omega_{N_{j}}^{-t_{j} n_{j}}, 
     \qquad t_{j}, n_{j} = -d_{j}, \dots, d_{j}.
\end{equation}
Then, since $N_{j} = 2d_{j} + 1$, each $\Phi^{(j)}$ satisfies
\begin{equation}
  \xpar[\big]{\amatrix{\Phi}^{(j)}}^{\A}\, \amatrix{\Phi}^{(j)} = N_{j}\, \identitymatrix_{N_{j}}.
\end{equation}
Thus, the multivariate evaluation tensor $\amatrix{P}$ can be obtained from the coefficient tensor $\amatrix{F}$ by applying the \gls{DFT} along each mode, i.e., by performing $M$ sequential batches of independent univariate \glspl{DFT} along each mode.
\begin{equation}
\label{eq:TensorDFT}
  \amatrix{P}
    = \amatrix{F} \times_{1} \amatrix{\Phi}^{(1)} \times_{2} \amatrix{\Phi}^{(2)} \dots \times_{M} \amatrix{\Phi}^{(M)}.
\end{equation}
Here $\times_{j}$ denotes the \emph{mode-$j$ product} of a tensor with a matrix, a generalization of matrix multiplication to tensors given by
\begin{equation}
  (\amatrix{F} \times_{j} \amatrix{\Phi}^{(j)})[t_{1}, \dots, t_{M}]
    \eqdef \sum_{n_{j} = -d_{j}}^{d_{j}} \amatrix{\Phi}^{(j)}_{t_{j}, n_{j}}\, \amatrix{F}[t_{1}, \dots, t_{j - 1}, n_{j}, t_{j + 1}, \dots, t_{M}],
\end{equation}
for all valid indices $t_{1}, \dots, t_{M}$.
In other words, to apply $\amatrix{\Phi}^{(j)}$ in mode $j$, we fix all indices except the \param{j}-th, view $\amatrix{F}$ along that direction as a length\nobreakdash-$N_{j}$ vector (a mode\nobreakdash-$j$ fiber), and multiply it by $\amatrix{\Phi}^{(j)}$.
In vectorized form, \eqref{eq:TensorDFT} is equivalent to
\begin{equation}
  \mathrm{vec}(\amatrix{P})
    = \xpar[\big]{\amatrix{\Phi}^{(M)} \otimes \dots \otimes \amatrix{\Phi}^{(1)}}\, \mathrm{vec}(\amatrix{F}).
\end{equation}

Using that $(\amatrix{\Phi}^{(j)})^{\A} \amatrix{\Phi}^{(j)} = N_{j} \identitymatrix_{N_{j}}$, it follows that the mapping in \eqref{eq:TensorDFT} is invertible; that is,
\begin{equation}
  \amatrix{F}
    = \amatrix{P} \times_{1} \xpar[\Big]{\frac{1}{N_{1}}\xpar[\big]{\amatrix{\Phi}^{(1)}}^{\A}} \times_{2} \xpar[\Big]{\frac{1}{N_{2}}\xpar[\big]{\amatrix{\Phi}^{(2)}}^{\A}} \dots \times_{M} \xpar[\Big]{\frac{1}{N_{M}}\xpar[\big]{\amatrix{\Phi}^{(M)}}^{\A}}.
\end{equation}
Here, each mode\nobreakdash-$j$ multiplication applies $\xpar[\big]{\amatrix{\Phi}^{(j)}}^{\A}$ to all mode\nobreakdash-$j$ fibers. 

To understand the computational structure, recall that multiplication by $\amatrix{\Phi}^{(j)}$ in mode $j$ acts only along mode\nobreakdash-$j$ fibers.  
For instance, in the case $j = 1$, let
\begin{equation}
  \amatrix{F}^{(1)} 
    \eqdef \amatrix{F} \times_{1} \amatrix{\Phi}^{(1)}.
\end{equation}
Then, by definition, we have that
\begin{equation}
  \amatrix{F}^{(1)}[t_{1}, \dots, t_{M}] 
    = \sum_{n_{1} = -d_{1}}^{d_{1}} \amatrix{\Phi}^{(1)}_{t_{1}, n_{1}}\, \amatrix{F}[n_{1}, t_{2}, \dots, t_{M}]
    \qquad \forall t_{1}, \dots, t_{M}.
\end{equation}
Thus, for fixed $(t_{2}, \dots, t_{M})$, the slice $\amatrix{F}[\colon\!, t_{2}, \dots, t_{M}]$ (with all indices fixed except the first one) is a mode\nobreakdash-$1$ fiber of length $N_{1}$, and the multiplication is exactly a 1D \gls{DFT} of length $N_{1}$.
Therefore, applying $\times_{1} \amatrix{\Phi}^{(1)}$ performs $N_{2} N_{3} \cdots N_{M}$ \textbf{independent} 1D \glspl{DFT}, one along each mode\nobreakdash-$1$ fiber. 
Since the \glspl{DFT} are independent, each 1D \gls{DFT} of length $N_{1}$ can be computed using the 1D \gls{FFT} in $\bigoh(N_{1} \log N_{1})$, and the total complexity of the mode\nobreakdash-$1$ step is
\begin{equation}
  \bigoh(N_{2} N_{3} \cdots N_{M}\, N_{1} \log N_{1}).
\end{equation}

Repeating this reasoning for modes $j = 2, \dots, M$, the complete inverse (or forward) \param{M}-dimensional \gls{DFT} can be computed by $M$ sequential batches of independent 1D \glspl{FFT}, yielding overall complexity
\begin{equation}
  \bigoh\xpar[\Big]{\sum_{j = 1}^{M} \prod_{i = 1}^{M} N_{i} \log N_{j}},
\end{equation}
or equivalently, by setting $N = N_{1} N_{2} \cdots N_{M}$, the total number of sampled points, we can rewrite the above complexity as 
\begin{equation}
  \bigoh(N \log N)
\end{equation}
with \emph{full parallelism} across fibers.

\section{Training poly-depth constant-parameter \glsfmtshortpl{PQC}}
\label{sec:GlobalPQCOptimization}

In \cref{sec:EfficientApproximationAlgorithm}, we present \emph{a hybrid quantum--classical algorithm} that, under suitable assumptions on the architecture of the \gls{PQC} and the observable (cf.~\cref{sec:Assumptions}), can efficiently (i.e., in time polynomial in $n$) approximate \emph{the global minimum} $f^{\star}$ up to polynomial accuracy with high probability, using only a polynomial number of queries to the quantum device. 
The algorithm relies on \pgls{HTP} approximation of the objective function $f$, obtained via the \gls{FFT}.
A sufficient condition (cf.~\cref{ass:IntegerSpectralDifferences}) for the existence of such a representation is that the generators of the \gls{PQC} have integer spectral differences.
As we have discussed in \cref{sec:Preliminaries}, \gls{HTP} optimization can be addressed using \glsdashedorshort{SOS} programming.
The main challenge of this approach lies in bounding the order of the \glsdashedorshort{SOS} relaxation required to approximate $f^{\star}$ up to polynomial accuracy.
This is addressed via the complexity guarantees in \cref{ass:PQCComplexity,ass:ObservableComplexity} discussed in \cref{sec:Assumptions}.

\paragraph{Notation.} We now introduce the necessary notation to state \cref{thm:HTPRepresentation,thm:EfficientApproximationAlgorithm,thm:OptimizerExtraction}.
For each $k \in \range{K}$, let $\aset{L}_{k} \subset \reals$ denote \emph{the spectrum} of the generator $\amatrix{V}_{k}$, and let $\Delta_{k} \eqdef \max \aset{L}_{k} - \min \aset{L}_{k}$ denote \emph{the spectral diameter} of $\amatrix{V}_{k}$.
We define the vector $\Delta \in \nnreals^{K}$ given by
\begin{equation}
  \Delta
    \eqdef (\Delta_{1}, \dots, \Delta_{K})
\end{equation}
to collect the spectral diameters of the generators of the \gls{PQC} $\amatrix{U}$.
Furthermore, for each $k \in \range{K}$ and $\lambda \in \aset{L}_{k}$, we define $\amatrix{P}_{k, \lambda} \from \aset{H} \to \aset{H}$ to be \emph{the orthogonal projector} onto the eigenspace of $\amatrix{V}_{k}$ corresponding to the eigenvalue $\lambda$.
We define $\aset{L} \subset \reals^{K}$ to be the \emph{set of frequency vectors} given by
\begin{equation}
  \aset{L}
    \equiv \aset{L}_{1} \times \dots \times \aset{L}_{K},
\end{equation}
and set 
\begin{equation}
  \Delta\aset{L} 
    \eqdef \aset{L} - \aset{L} 
    \equiv \set{\lambda_{+} - \lambda_{-} \given \lambda_{+}, \lambda_{-} \in \aset{L}}
\end{equation}
to be the set of spectral differences.
Finally, for each $\lambda \in \aset{L}$, we define the operator $\amatrix{P}_{\lambda} \from \aset{H} \to \aset{H}$ as
\begin{equation}
  \amatrix{P}_{\lambda}
    \eqdef \amatrix{P}_{K, \lambda_{K}} \amatrix{C}_{K} \cdots \amatrix{P}_{1, \lambda_{1}} \amatrix{C}_{1}.
\end{equation}

It is also convenient to encode the \gls{PQC}['s] \emph{parameter sharing} among the gates $\amatrix{U}_{1}$, \dots, $\amatrix{U}_{K}$ using \emph{a binary matrix} $\amatrix{A} \in \set{0, 1}^{K \times M}$.
We define $\amatrix{A}$ such that its \param{j}-th column $\amatrix{A}_{\colon\!, j}$ indicates which parametrized gates $\amatrix{U}_{k}$ depend on the independent parameter $\theta_{j}$, and its \param{k}-th row $\amatrix{A}_{k, \colon\!}$ indicates which independent parameter $\theta_{j}$ parametrizes the gate $\amatrix{U}_{k}$.
In particular, the \param{k}-th row $\amatrix{A}_{k, \colon\!} \in \set{0, 1}^{M}$ of $\amatrix{A}$ has a single nonzero entry at position $j_{k}$ such that
\begin{equation}
  \amatrix{A}_{k, \colon\!}\, \theta
    = \theta_{j_{k}},
    \qquad \forall k \in \range{K}.
\end{equation}

\subsection{Representing the \glsfmtshort{PQC} objective as \pglsfmtlong{HTP}}
\label{sec:HTPRepresentation}

The following \lcnamecref{thm:HTPRepresentation} characterizes the objective functions induced by \glspl{PQC} whose generators have integer spectral differences.
We remark that, although related representation results have been obtained in the literature~\citep{schuld_2021_effect,fontana_2022_efficient,nemkov_2023_fourier}, \cref{thm:HTPRepresentation} makes explicit how the max degree and coefficients of the \gls{HTP} representation depend on the architecture of the \gls{PQC} $\amatrix{U}$, and in particular, the spectra of the generators $\amatrix{V}_{1}, \dots, \amatrix{V}_{K}$ and the parameter-sharing matrix $\amatrix{A}$.
We use this explicit dependence in \cref{sec:EfficientApproximationAlgorithm} to bound the order of the \glsdashedorshort{SOS} relaxation needed to approximate $f^{\star}$.

\begin{theorem}
\label{thm:HTPRepresentation}
  Consider \pgls{PQC} $\amatrix{U}$ that satisfies \cref{ass:IntegerSpectralDifferences}, i.e., the spectra $\aset{L}_{1}$, \dots, $\aset{L}_{K}$ of, respectively, the generators $\amatrix{V}_{1}$, \dots, $\amatrix{V}_{K}$ have integer differences. 
  Then, for every Hermitian observable $\amatrix{O} \from \aset{H} \to \aset{H}$, the objective function $f \from \reals^{M} \to \reals$ is \pgls{HTP} of the form
  \begin{equation}
    f(\theta)
      = \sum_{\alpha \in \amatrix{A}^{\T}(\Delta\aset{L})} f_{\alpha} \exp\xpar[\big]{\im \inner{\alpha}{\theta}},
  \ \text{where} \
    f_{\alpha}
      = \sum_{\substack{
        \lambda_{+}, \lambda_{-} \in \aset{L} \\ 
        \amatrix{A}^{\T}(\lambda_{+} - \lambda_{-}) = \alpha
      }} \braket{0}{\amatrix{P}_{\lambda_{+}}^{\A} \amatrix{O} \amatrix{P}_{\lambda_{-}}}{0},
    \qquad \forall \alpha \in \amatrix{A}^{\T}(\Delta\aset{L}),
  \end{equation}
  Moreover, for each independent parameter $\theta_{j}$, we have $\abs{\alpha_{j}} \leq (\amatrix{A}^{\T} \Delta)_{j}$, i.e., the sum of spectral diameters of all generators that depend on $\theta_{j}$ is bounded by $(\amatrix{A}^{\T} \Delta)_{j}$.
\end{theorem}
\begin{proof}
  Since every real-valued trigonometric polynomial is \pgls{HTP} and the objective function $f$ is real-valued, it suffices to show that $f$ is a trigonometric polynomial with max degree at most $d_{j}$ for each independent parameter $\theta_{j}$.

  \paragraph{Step 1: Reformulation as \pgls{HTP}.} To begin with, consider an arbitrary $k \in \range{K}$ and observe that, since $\amatrix{V}_{k}$ is Hermitian, it admits the spectral decomposition
  \begin{equation}
    \amatrix{V}_{k}
      = \sum_{\lambda \in \aset{L}_{k}} \lambda \amatrix{P}_{k, \lambda}.
  \end{equation}
  Hence, for each $\theta_{j_{k}} \in \reals$,
  \begin{equation}
    \exp\xpar[\big]{-\im \theta_{j_{k}} \amatrix{V}_{k}}
      = \sum_{\lambda \in \aset{L}_{k}} \exp(-\im \lambda \theta_{j_{k}}) \amatrix{P}_{k, \lambda}.
  \end{equation}

  Next, consider an arbitrary $\theta \in \reals^{M}$.
  Substituting the above expression for each factor $\amatrix{U}_{k}$ into the definition of the \gls{PQC} $\amatrix{U}$, we obtain
  \begin{subequations}
  \begin{align}
    \amatrix{U}(\theta)
      &= \amatrix{U}_{K}\xpar[\big]{\amatrix{A}_{K, \colon\!} \theta} \amatrix{C}_{K} \dots \amatrix{U}_{1}\xpar[\big]{\amatrix{A}_{1, \colon\!} \theta} \amatrix{C}_{1} \\
      &= \sum_{\lambda_{K} \in \aset{L}_{K}} \exp\xpar[\big]{-\im \lambda_{K} \amatrix{A}_{K, \colon\!} \theta} \amatrix{P}_{K, \lambda_{K}} \amatrix{C}_{K} \dots \sum_{\lambda_{1} \in \aset{L}_{1}} \exp\xpar[\big]{-\im \lambda_{1} \amatrix{A}_{1, \colon\!} \theta} \amatrix{P}_{1, \lambda_{1}} \amatrix{C}_{1} \\
      &= \sum_{\lambda \in \aset{L}} \exp\xpar[\bigg]{-\im \sum_{k  = 1}^{K} \lambda_{k} \amatrix{A}_{k, \colon\!} \theta} \amatrix{P}_{\lambda} \\
      &= \sum_{\lambda \in \aset{L}} \exp\xpar[\big]{-\im \inner{\amatrix{A}^{\T} \lambda}{\theta}} \amatrix{P}_{\lambda}.
  \end{align}
  \end{subequations}
  Therefore, by setting $\beta \equiv \amatrix{A}^{\T} \lambda$ and grouping the terms accordingly, we obtain
  \begin{equation}
  \label{eq:PQCExpansion}
    \amatrix{U}(\theta)
      = \sum_{\beta \in \amatrix{A}^{\T} \aset{L}} \exp\xpar[\big]{-\im \inner{\beta}{\theta}} \sum_{\substack{
        \lambda \in \aset{L} \\
        \amatrix{A}^{\T} \lambda = \beta
      }} \amatrix{P}_{\lambda}.
  \end{equation}
  Then the adjoint $\amatrix{U}^{\A}(\theta)$ can be expressed, by linearity of the adjoint, as
  \begin{equation}
  \label{eq:PQCAdjointExpansion}
    \amatrix{U}^{\A}(\theta)
      = \sum_{\beta \in \amatrix{A}^{\T} \aset{L}} \conj{\exp\xpar[\big]{ -\im \inner{\beta}{\theta}}} \sum_{\substack{
        \lambda \in \aset{L} \\
        \amatrix{A}^{\T} \lambda = \beta
      }} \amatrix{P}_{\lambda}^{\A}
      = \sum_{\beta \in \amatrix{A}^{\T} \aset{L}} \exp\xpar[\big]{\im \inner{\beta}{\theta}} \sum_{\substack{
        \lambda \in \aset{L} \\
        \amatrix{A}^{\T} \lambda = \beta
      }} \amatrix{P}_{\lambda}^{\A}.
  \end{equation}
  Consequently, substituting \cref{eq:PQCExpansion,eq:PQCAdjointExpansion} into the definition of the objective function $f$ yields
  \begin{subequations}
  \begin{align}
    f(\theta)
      &= \braket[\big]{0}{\amatrix{U}^{\A}(\theta) \amatrix{O} \amatrix{U}(\theta)}{0} \\
      &= \braket[\Bigg]{0}{\xpar[\bigg]{\sum_{\beta_{+} \in \amatrix{A}^{\T} \aset{L}} \exp\xpar[\big]{\im \inner{\beta_{+}}{\theta}} \sum_{\substack{
        \lambda_{+} \in \aset{L} \\
        \mathclap{\amatrix{A}^{\T} \lambda_{+} = \beta_{+}}
      }} \amatrix{P}_{\lambda_{+}}^{\A}} \amatrix{O} \xpar[\bigg]{\sum_{\beta_{-} \in \amatrix{A}^{\T} \aset{L}} \exp\xpar[\big]{-\im \inner{\beta_{-}}{\theta}} \sum_{\substack{
        \lambda_{-} \in \aset{L} \\
        \mathclap{\amatrix{A}^{\T} \lambda_{-} = \beta_{-}}
      }} \amatrix{P}_{\lambda_{-}}}}{0} \\
      &= \sum_{\beta_{+}, \beta_{-} \in \amatrix{A}^{\T} \aset{L}} \exp\xpar[\Big]{\im \inner[\big]{(\beta_{+} - \beta_{-})}{\theta}} \sum_{\substack{
        \lambda_{+}, \lambda_{-} \in \aset{L} \\ 
        \amatrix{A}^{\T} \lambda_{+} = \beta_{+} \\ 
        \amatrix{A}^{\T} \lambda_{-} = \beta_{-}
      }} \braket{0}{\amatrix{P}_{\lambda_{+}}^{\A} \amatrix{O} \amatrix{P}_{\lambda_{-}}}{0} \\
  \end{align}
  \end{subequations}
  
  Finally, by setting $\alpha \equiv \beta_{+} - \beta_{-}$ and grouping the terms accordingly, we obtain
  \begin{equation}
    f(\theta)
      = \sum_{\alpha \in \amatrix{A}^{\T}(\Delta\aset{L})} f_{\alpha} \exp\xpar[\big]{\im \inner{\alpha}{\theta}}.
  \end{equation}
  So the coefficient $f_{\alpha}$ collects all pairs of spectral contributions whose projected difference equals $\alpha$.
  Thus, since the entries of $\amatrix{A}$ are integers and, by \cref{ass:IntegerSpectralDifferences}, the spectra $\aset{L}_{1}, \dots, \aset{L}_{K}$ have integer differences, we conclude that the frequencies $\alpha \in \amatrix{A}^{\T}(\Delta\aset{L})$ of $f$ are integers; therefore, $f$ is a trigonometric polynomial.

  \paragraph{Step 2: Degree bound.} Now consider an arbitrary independent parameter $\theta_{j}$.
  Then, since all entries of $\amatrix{A}$ are nonnegative, we have that
  \begin{subequations}
  \begin{align}
    \max \xpar[\big]{\amatrix{A}^{\T} (\Delta\aset{L})}_{j}
      &= \max_{\alpha \in \Delta\aset{L}} \inner{\amatrix{A}_{\colon\!, j}}{\alpha} \\ 
      &= \max_{\alpha \in \Delta\aset{L}} \sum_{k = 1}^{K} \amatrix{A}_{k, j} \alpha_{k} \\
      &\leq \sum_{k = 1}^{K} \amatrix{A}_{k, j} \max (\aset{L}_{k} - \aset{L}_{k}) \\
      &= \sum_{k = 1}^{K} \amatrix{A}_{k, j} (\max \aset{L}_{k} - \min \aset{L}_{k}) \\
      &= \sum_{k = 1}^{K} \amatrix{A}_{k, j} \Delta_{k} \\
      &= (\amatrix{A}^{\T} \Delta)_{j}.
  \end{align}
  \end{subequations}
  Furthermore, since $\Delta\aset{L} = \aset{L} - \aset{L}$ is symmetric around zero, we also have
  \begin{equation}
    \min \xpar[\big]{\amatrix{A}^{\T} (\Delta\aset{L})}_{j}
      = - \max \xpar[\big]{\amatrix{A}^{\T} (\Delta\aset{L})}_{j}
      \geq - (\amatrix{A}^{\T} \Delta)_{j}.
  \end{equation}
  Therefore, all frequencies in the \param{j}-th coordinate lie in the interval $\ccinterval[\big]{-(\amatrix{A}^{\T} \Delta)_{j}}{(\amatrix{A}^{\T} \Delta)_{j}}$, and thus, the max degree of $f$ with respect to the independent parameter $\theta_{j}$ is at most $(\amatrix{A}^{\T} \Delta)_{j}$.
  This explicit degree formula is what later controls sampling grid size and the \glsdashedorshort{SOS} order in \cref{sec:EfficientApproximationAlgorithm}.
\end{proof}


\subsection{A range-based additive \glsfmtshort{FPRAS} for \glsfmtshort{PQC} optimization}
\label{sec:EfficientApproximationAlgorithm}

We now present a hybrid quantum--classical algorithm for efficiently approximating the optimal value of the \gls{PQC} objective function.
The \lcnamecref{alg:EfficientApproximationAlgorithm} is a \glsxtrfull{FPRAS} in the weak sense (additive) for approximating $f^{\star}$, and requires only a polynomial number of queries to the quantum device.
The main result of this \lcnamecref{sec:EfficientApproximationAlgorithm} is stated in \cref{thm:EfficientApproximationAlgorithm}.

\begin{theorem}[\Pglsfmtshort{FPRAS} for \glsfmtshort{PQC} optimization]
\label{thm:EfficientApproximationAlgorithm}
  Consider \pgls{PQC} $\amatrix{U}$ with $n$ qubits that satisfies \cref{ass:IntegerSpectralDifferences,ass:PQCComplexity}.
  Let $\amatrix{O} \from \aset{H} \to \aset{H}$ be a Hermitian observable that satisfies \cref{ass:ObservableComplexity}, i.e., its operator norm $\norm{\amatrix{O}}$ is bounded by a polynomial in $n$.
  Then, for every error $\epsilon > 0$ and failure probability $\delta > 0$, there exists a randomized algorithm with oracle access to a quantum device running in time polynomial in $n$, $1 / \epsilon$, and $\log (1 / \delta)$ that outputs an estimate $\hat{f}^{\star}$ of $f^{\star}$ such that, with probability at least $1 - \delta$,
  \begin{equation}
    \abs{\hat{f}^{\star} - f^{\star}} \leq \epsilon.
  \end{equation}
  Furthermore, the algorithm requires a number of queries to the quantum device that is a priori known and polynomial in $n$, $1 / \epsilon$, and $\log (1 / \delta)$.
\end{theorem}

\paragraph{Proof sketch of \cref{thm:EfficientApproximationAlgorithm}.} Since the formal proof contains a few technical details, before we present it, we begin by providing a high-level sketch of it for clarity.
Let us begin by fixing the bounds provided by our \lcnamecrefs{ass:IntegerSpectralDifferences}, so they do not appear as a surprise later on.
The bounds we use throughout are:
\begin{itemize}
  \item $M$ is constant~(\cref{ass:IndependentParameterCount}).
  \item $K$ is polynomial in $n$~(\cref{ass:ParametrizedGatesCount}).
  \item The generator spectral diameters are bounded by $\Delta_{\max} = \polytime(n)$~(\cref{ass:GeneratorSpectralDiameterBound}).
  \item The operator norm of the observable is bounded by $O_{\max} = \polytime(n)$~(\cref{ass:ObservableComplexity}).
\end{itemize}
The proof consists of a few concrete steps.

\paragraph{Step 1: Representing the objective $f$ as \pgls{HTP}, and bounding its frequencies.}
By \cref{thm:HTPRepresentation}, the objective $f$ is \pgls{HTP}, and each coordinate degree is at most $d_{j} = (\amatrix{A}^{\T} \Delta)_{j}$. 
Using the bounds above, this implies all frequencies of $f$ lie in a \emph{finite box}
\begin{equation}
  \aset{K} 
    = [-K \Delta_{\max}, K \Delta_{\max}]^{M},
\end{equation}
so $\abs{\aset{K}} = (2 K\, \Delta_{\max} + 1)^{M}$, which is polynomial in $n$ because $M$ is constant, and $K$ and $\Delta_{\max}$ are polynomial in $n$.

\paragraph{Step 2: Sampling the quantum device and bounding the error.} We sample on \emph{a matching uniform grid} $\aset{S} \subset \ccinterval{-\pi}{\pi}^{M}$ with $\abs{\aset{S}} = \abs{\aset{K}}$.
For each grid point $\phi \in \aset{S}$, we perform $N$ shots on the quantum hardware, i.e., prepare the state $\amatrix{U}(\phi) \ket{0}$, and average $N$ measurements of the observable $\amatrix{O}$.

Here we use the standard Hoeffding concentration bound for quantum expectation estimation with bounded outcomes~\citep{hoeffding_1963_probability}.
In particular, if the outcomes lie in an interval of length $R$, then Hoeffding's inequality ensures that, for each $\phi \in \aset{S}$,
\begin{equation}
  \prob[\Big]{\abs[\big]{\tilde{f}(\phi) - f(\phi)} \leq \varepsilon} 
    \geq 1 - 2 \exp\xpar[\Big]{-\frac{2 N \varepsilon^{2}}{R^{2}}}.
\end{equation}
In our case $R \leq 2 O_{\max}$. 
Therefore, choosing
\begin{equation}
\label{eq:PerPointMeasurementCount}
  \varepsilon 
    = \frac{\epsilon}{2 \abs{\aset{S}}},
  \qquad
  N 
   = \frac{8 \abs{\aset{S}}^{2} O_{\max}^{2}}{\epsilon^{2}} \log \frac{2 \abs{\aset{S}}}{\delta}
\end{equation}
gives per-point failure probability at most $\delta / \abs{\aset{S}}$.
Then a union bound gives simultaneous control over all $\phi \in \aset{S}$ with probability at least $1 - \delta$, and therefore
\begin{equation}
  \prob[\bigg]{\abs[\big]{\tilde{f}(\phi) - f(\phi)} \leq \frac{\epsilon}{2 \abs{\aset{S}}}, \,\forall \phi \in \aset{S}} 
    \geq 1 - \delta.
\end{equation}

\paragraph{Step 3: Interpolating via \gls{FFT} and bounding the approximation error.} From the sampled values, we compute Fourier coefficients using \gls{FFT} as in \cref{sec:TrigonometricPolynomialsFFT}, then enforce Hermitian symmetry to form the approximation $\hat{f}$ of the objective $f$.
Under the high-probability event above, coefficient errors are uniformly small; summing them gives
\begin{equation}
  \sup_{\theta} \abs[\big]{f(\theta) - \hat{f}(\theta)} 
    \leq \frac{\epsilon}{2}.
\end{equation}
So their optima differ by at most $\epsilon / 2$, i.e.,
\begin{equation}
  \abs[\big]{f^{\star} - \hat{f}^{\star}} 
    \leq \frac{\epsilon}{2}.
\end{equation}

\paragraph{Step 4: Constructing the trigonometric \glsdashedorshort{SOS} hierarchy and bounding the per-level optimization error.}
We then optimize $\hat{f}$ with the max-degree trigonometric \glsdashedorshort{SOS} hierarchy whose minimum at level $\ell$ is $\hat{f}_{\ell}^{\star}$.
The bound of \citet{bach_2023_exponential} controls $\abs{\hat{f}^{\star} - \hat{f}_{\ell}^{\star}}$ in terms of $\ell$.
Then, using the coefficient bound derived in the proof, we choose $\ell = \polytime(n, 1 / \epsilon)$ so that 
\begin{equation}
  \abs[\big]{\hat{f}^{\star} - \hat{f}_{\ell}^{\star}}
    \leq \frac{\epsilon}{2}.
\end{equation}

\paragraph{Step 5: Computing the final guarantee and complexity.} Combining both halves,
\begin{equation}
  \abs[\big]{f^{\star} - \hat{f}_{\ell}^{\star}}
    \le \abs[\big]{f^{\star} - \hat{f}^{\star}} + \abs[\big]{\hat{f}^{\star} - \hat{f}_{\ell}^{\star}}
    \le \epsilon.
\end{equation}
This holds with probability at least $1 - \delta$ as we conditioned on the event of the union bound in Step 2.
Moreover, the runtime of the algorithm is polynomial in $n$, $1 / \epsilon$, and $\log(1 / \delta)$ by construction, since \gls{FFT} is quasi-linear in $N$, and the \gls{SOS} optimization is polynomial because $M$ is constant and $\ell$ is polynomial in $n$ and $1 / \epsilon$.

The proof of \cref{thm:EfficientApproximationAlgorithm} is a constructive one.
In particular, it specifies the steps of the hybrid quantum--classical algorithm summarized in \cref{alg:EfficientApproximationAlgorithm}, and shows that, under the assumptions of \cref{thm:EfficientApproximationAlgorithm}, these steps can be implemented in time polynomial in $n$, $1 / \epsilon$, and $\log (1 / \delta)$.
Note that \cref{alg:EfficientApproximationAlgorithm} is a range-based additive \gls{FPRAS}.

\begin{algorithm}[H]
  \SetKwInOut{Input}{input}\SetKwInOut{Output}{output}
  \DontPrintSemicolon
  \Input{An error $\epsilon > 0$, a failure probability $\delta > 0$}
  \Output{An approximation $\hat{f}^{\star}$ such that $\abs{\hat{f}^{\star} - f^{\star}} \le \epsilon$ with probability at least $1 - \delta$}
  \BlankLine
  Construct a sampling grid $\aset{S} \subset \ccinterval{-\pi}{\pi}^{M}$ with $2 K \Delta_{\max} + 1$ equispaced points in each coordinate $j$\;
  Set the per-sample shot count $N$ as in \cref{eq:PerPointMeasurementCount}\;
  \For{$\phi \in \aset{S}$}{
    Prepare the state $\amatrix{U}(\phi)\ket{0}$ and perform $N$ measurements $X(\phi)$ of the observable $\amatrix{O}$\;
    Let $\tilde{f}(\phi)$ be the empirical mean $X(\phi)$ over the $N$ shots\;
  }
  Compute an approximation $\hat{f}$ of $f$ using \gls{FFT} on the empirical means $\set[\big]{\tilde{f}(\phi) \given \phi \in \aset{S}}$\;
  Set the \glsdashedorshort{SOS} level $\ell$ as in \cref{eq:RequiredSOSLevel}\;
  \Return the solution of the max-degree trigonometric \glsdashedorshort{SOS} hierarchy for $\min_{\theta} \hat{f}(\theta)$ at level $\ell$\;
  \caption{Hybrid quantum--classical algorithm for approximating $f^{\star}$}
  \label{alg:EfficientApproximationAlgorithm}
\end{algorithm}

\subsubsection{Proof of the additive \glsfmtshort{FPRAS}}
\label{sec:EfficientApproximationAlgorithmProof}

We now present the formal proof of \cref{thm:EfficientApproximationAlgorithm}.
At a high level, the proof splits the target error budget $\epsilon$ into two parts, and spends $\epsilon / 2$ for constructing an approximation $\hat{f}$ to the objective function $f$, and another $\epsilon / 2$ for approximating the optimal value of $\min_{\theta \in \ccinterval{-\pi}{\pi}^{M}}\hat{f}$ via a max-degree trigonometric \glsdashedorshort{SOS} hierarchy.

\paragraph{Step 1: Representing the objective $f$ as \pgls{HTP}, and bounding its frequencies.} Since the spectra of the generators $\amatrix{V}_{1}, \dots, \amatrix{V}_{K}$ have integer differences~(\cref{ass:IntegerSpectralDifferences}), by \cref{thm:HTPRepresentation}, it follows that the objective function $f$ is \pgls{HTP} with max degree at most $d_{j} = (\amatrix{A}^{\T} \Delta)_{j}$ for each independent parameter $\theta_{j}$.
Then, by the assumed bounds, for each $j \in \range{M}$, we have
\begin{equation}
  d_{j} 
    = (\amatrix{A}^{\T} \Delta)_{j}
    \leq \norm{\amatrix{A}_{\colon\!, j}}_{1}\, \max_{k \in \range{K}} \Delta_{k}
    \leq K \Delta_{\max}.
\end{equation}

We define the frequency set
\begin{equation}
  \aset{K} 
    \eqdef \ccinterval{-K \Delta_{\max}}{K \Delta_{\max}}^{M}.
\end{equation}
Since $\bigtimes_{j = 1}^{M} \ccinterval{-d_{j}}{d_{j}} \subseteq \aset{K}$, it follows that $\aset{K}$ contains all possible frequencies of $f$.
Intuitively, we \emph{over-approximate} the true frequency support by a simple hypercube $\aset{K}$ because this yields a uniform sampling grid for the \gls{FFT} interpolation in the following step.
Moreover, since $M$ is constant in $n$~(\cref{ass:IndependentParameterCount}), $K$ is polynomial in $n$~(\cref{ass:ParametrizedGatesCount}), and $\Delta_{\max}$ is polynomial in $n$~(\cref{ass:GeneratorSpectralDiameterBound}), it follows that $\abs{\aset{K}} = (2 K \Delta_{\max} + 1)^{M}$ is a polynomial in $n$.

\paragraph{Step 2: Sampling the quantum device and bounding the error.} Consider the sampling grid $\aset{S} \subset \ccinterval{-\pi}{\pi}^{M}$ obtained by taking $2 K \Delta_{\max} + 1$ equispaced points along each independent parameter $\theta_{j}$.
Then $\abs{\aset{S}} = \abs{\aset{K}}$.
Matching $\abs{\aset{S}}$ with $\abs{\aset{K}}$ ensures that the \gls{FFT} interpolation can be applied to the constructed sample in the subsequent step.

For each sampling point $\phi \in \aset{S}$, we compute an estimate of $f(\phi)$ by performing 
\begin{equation}
  N
    = \frac{8 \abs{\aset{S}}^{2} O_{\max}^{2}}{\epsilon^{2}} \ln\xpar[\bigg]{\frac{2 \abs{\aset{S}}}{\delta}}
\end{equation}
independent measurements $X(\phi)$ of the observable $\amatrix{O}$ over the prepared quantum state $\amatrix{U}(\phi) \ket{0^{\otimes n}}$.
Since $\abs{\aset{S}}$ and $O_{\max}$ are polynomial in $n$~(\cref{ass:ObservableComplexity}), the number of shots $N$ is polynomial in $n$, $1 / \epsilon$, and $\log (1 / \delta)$.
Moreover, since $\abs{\aset{S}} = (2 K \Delta_{\max} + 1)^{M}$, the overall number of queries to the quantum device is \emph{exactly}
\begin{equation}
  N\, \abs{\aset{S}}
    = N\, (2 K \Delta_{\max} + 1)^{M},
\end{equation}
which is also polynomial in $n$, $1 / \epsilon$, and $\log (1 / \delta)$.

Now, let $\tilde{f}(\phi)$ denote the empirical mean of $X(\phi)$ over these $N$ shots.
Each measurement outcome lies in an interval of length $2 O_{\max} \geq 2 \norm{\amatrix{O}} \geq \Delta_{\amatrix{O}}$, i.e., the spectral diameter of $\amatrix{O}$, so Hoeffding's inequality with range $2 O_{\max}$ implies that, for each $\phi \in \aset{S}$,
\begin{equation}
  \prob[\bigg]{\abs[\big]{\tilde{f}(\phi) - f(\phi)} \leq \frac{\epsilon}{2 \abs{\aset{S}}}} 
    \geq 1 - \frac{\delta}{\abs{\aset{S}}}.
\end{equation}
Therefore, by applying the union bound, we obtain that
\begin{equation}
\label{eq:UnionBound}
  \prob[\bigg]{\abs[\big]{\tilde{f}(\phi) - f(\phi)} \leq \frac{\epsilon}{2 \abs{\aset{S}}}, \,\forall \phi \in \aset{S}} 
    \geq 1 - \delta.
\end{equation}
In the remainder of the proof we condition on this event, which occurs with probability at least $1 - \delta$, and therefore, from this point onwards, all bounds are deterministic on the high-probability event in \eqref{eq:UnionBound}.

\paragraph{Step 3: Interpolating via \gls{FFT} and bounding the approximation error.} We compute the \gls{DFT} of the estimates $\tilde{f}(\phi)$ over the sampling grid $\aset{S}$ to obtain estimates
\begin{equation}
  \tilde{f}_{\alpha}
    = \frac{1}{\abs{\aset{S}}} \sum_{\phi \in \aset{S}} \tilde{f}(\phi) \exp\xpar[\big]{-\im \inner{\alpha}{\phi}}
\end{equation}
of the coefficients $f_{\alpha}$ of the objective function $f$ for each frequency $\alpha \in \aset{K}$.
Since the complexity of computing the \gls{DFT} with the \gls{FFT} algorithm on a sampling grid of size $\abs{\aset{S}}$ is $\bigoh\xpar[\big]{\abs{\aset{S}} \log \abs{\aset{S}}}$ (cf.~\cref{sec:TrigonometricPolynomialsFFT}) and $\abs{\aset{S}} = \abs{\aset{K}}$ is a polynomial in $n$, the estimates $\tilde{f}_{\alpha}$, for each $\alpha \in \aset{K}$, can be computed in polynomial time in $n$, $1 / \epsilon$, and $\log (1 / \delta)$.
Notice that $\tilde{f}_{\alpha}$ do not necessarily satisfy the Hermitian symmetry condition of \pgls{HTP}, i.e., $\tilde{f}_{-\alpha}$ may not be the complex conjugate of $\tilde{f}_{\alpha}$.

Define the approximate objective function $\hat{f} \from \reals^{M} \to \reals$ as
\begin{equation}
  \hat{f}(\theta)
    = \sum_{\alpha \in \aset{K}} \hat{f}_{\alpha} \exp\xpar[\big]{\im \inner{\alpha}{\theta}},
\ \text{where} \
  \hat{f}_{\alpha}
    \eqdef \frac{1}{2}(\tilde{f}_{\alpha} + \conj{\tilde{f}_{-\alpha}}),
  \qquad \forall \alpha \in \aset{K}.
\end{equation}
Observe that, by construction, $\hat{f}$ is \pgls{HTP} because 
\begin{equation}
  \hat{f}_{-\alpha}
    = \frac{1}{2}(\tilde{f}_{-\alpha} + \conj{\tilde{f}_{\alpha}}) 
    = \conj{\hat{f}_{\alpha}}.
\end{equation}
This symmetrization step restores the Hermitian symmetry condition without increasing the coefficient errors beyond constant factors.
Furthermore, since $\aset{K} = \ccinterval{-K \Delta_{\max}}{K \Delta_{\max}}^{M}$, it follows that $\hat{f}$ also has max degree at most $K \Delta_{\max}$ for each independent parameter $\theta_{j}$.

Next, we compute the approximation error of $\hat{f}$.
In particular, consider an arbitrary frequency $\alpha \in \aset{K}$.
By the \gls{DFT} over the sampling set $\aset{S}$, we have that
\begin{equation}
  f_{\alpha} 
    = \frac{1}{\abs{\aset{S}}} \sum_{\phi \in \aset{S}} f(\phi) \exp\xpar[\big]{-\im \inner{\alpha}{\phi}}.
\end{equation}
Furthermore, since $\exp\xpar[\big]{-\im \inner{\alpha}{\phi}}$ lies on the unit complex torus, for all $\phi \in \aset{S}$, it follows that $\abs{\exp\xpar[\big]{-\im \inner{\alpha}{\phi}}} = 1$.
Therefore, by the triangle inequality, we obtain
\begin{subequations}
\begin{align}
  \abs{\tilde{f}_{\alpha} - f_{\alpha}}
    &= \frac{1}{\abs{\aset{S}}} \abs[\Bigg]{\sum_{\phi \in \aset{S}} \xpar[\big]{\tilde{f}(\phi) - f(\phi)} \exp\xpar[\big]{-\im \inner{\alpha}{\phi}}} \\
    &\leq \frac{1}{\abs{\aset{S}}} \sum_{\phi \in \aset{S}} \abs{\tilde{f}(\phi) - f(\phi)} \\
    &\leq \frac{1}{\abs{\aset{S}}} \sum_{\phi \in \aset{S}} \frac{\epsilon}{2 \abs{\aset{K}}} \\
    &= \frac{\epsilon}{2 \abs{\aset{K}}}.
\end{align}
\end{subequations}
Moreover, since $f$ is \pgls{HTP}, i.e., $f_{\alpha} = \conj{f_{-\alpha}}$, by the triangle inequality, we also have
\begin{equation}
  \abs{\hat{f}_{\alpha} - f_{\alpha}}
    \leq \frac{1}{2} \xpar[\big]{\abs{\tilde{f}_{\alpha} - f_{\alpha}} + \abs{\conj{\tilde{f}_{-\alpha}} - \conj{f_{-\alpha}}}}
    = \frac{1}{2} \xpar[\big]{\abs{\tilde{f}_{\alpha} - f_{\alpha}} + \abs{\tilde{f}_{-\alpha} - f_{-\alpha}}}
    \leq \frac{\epsilon}{2 \abs{\aset{K}}}.
\end{equation}
So each frequency coefficient is controlled at scale $\epsilon / \abs{\aset{K}}$, which is exactly the normalization needed for bounding the \param{\ell_{1}}-norm of the approximation error.
Consequently, we obtain
\begin{equation}
\label{eq:ObjectiveFunctionErrorBound}
  \abs[\big]{\hat{f}(\theta) - f(\theta)}
    = \abs[\Bigg]{\sum_{\alpha \in \aset{K}} (\hat{f}_{\alpha} - f_{\alpha}) \exp\xpar[\big]{\im \inner{\alpha}{\theta}}}
    \leq \sum_{\alpha \in \aset{K}} \abs{\hat{f}_{\alpha} - f_{\alpha}}
    \leq \sum_{\alpha \in \aset{K}} \frac{\epsilon}{2 \abs{\aset{K}}}
    = \frac{\epsilon}{2},
    \qquad \forall \theta \in \reals^{M}.
\end{equation}
Importantly, the error bound above holds uniformly over all $\theta \in \reals^{M}$.

Now, define 
\begin{equation}
  \hat{f}^{\star} 
    \eqdef \min_{\theta \in \ccinterval{-\pi}{\pi}^{M}} \hat{f}(\theta)
\end{equation}
to be the minimum value of the approximate objective function $\hat{f}$, and recall that since $f$ is \param{2 \pi}-periodic along each independent parameter $\theta_{j}$, we also have that 
\begin{equation}
  f^{\star} 
    = \min_{\theta \in \ccinterval{-\pi}{\pi}^{M}} f(\theta).
\end{equation}
Thus, since $\ccinterval{-\pi}{\pi}^{M}$ is compact and the functions $f$ and $\hat{f}$ are continuous, there exist some minimizers $\theta^{\star}$ and $\hat{\theta}^{\star}$ of $f$ and $\hat{f}$, respectively, in $\ccinterval{-\pi}{\pi}^{M}$.
Then, by \eqref{eq:ObjectiveFunctionErrorBound}, it follows that
\begin{equation}
  \hat{f}^{\star} - f^{\star}
    \leq \hat{f}(\theta^{\star}) - f(\theta^{\star})
    \leq \abs[\big]{\hat{f}(\theta^{\star}) - f(\theta^{\star})}
    \leq \frac{\epsilon}{2},
\end{equation}
and, similarly,
\begin{equation}
  f^{\star} - \hat{f}^{\star}
    \leq f(\hat{\theta}^{\star}) - \hat{f}(\hat{\theta}^{\star})
    \leq \abs[\big]{f(\hat{\theta}^{\star}) - \hat{f}(\hat{\theta}^{\star})}
    \leq \frac{\epsilon}{2}.
\end{equation}
Therefore, we have that
\begin{equation}
\label{eq:GlobalOptimumErrorBound}
  \abs{\hat{f}^{\star} - f^{\star}} \leq \frac{\epsilon}{2}.
\end{equation}
In other words, the uniform approximation error of $\hat{f}$ to $f$ translates into a global optimum approximation error of $\hat{f}^{\star}$ to $f^{\star}$.

\paragraph{Step 4: Constructing a trigonometric \glsdashedorshort{SOS} hierarchy and bounding the per-level optimization error.} Consider the hierarchy of max-degree \glsdashedorshort{SOS} relaxations of the optimization problem 
\begin{equation}
  \min_{\theta \in \ccinterval{-\pi}{\pi}^{M}} \hat{f}(\theta).
\end{equation}
Let $\hat{f}^{\star}_{\ell}$ denote the optimal value of the relaxation at order $\ell \in \naturals$.
Then, by \citet[Theorem~3.1]{bach_2023_exponential}, we have that for all 
\begin{equation}
\label{eq:MinimumSOSLevel}
  \ell 
    > \ell_{0}
    \eqdef \sqrt{6} \ceil[\big]{\max_{j \in \range{M}} K \Delta_{\max} / 2}, 
\end{equation}
it holds that
\begin{equation}
  \abs{\hat{f}^{\star} - \hat{f}_{\ell}^{\star}}
    \leq \xpar[\bigg]{\xpar[\Big]{1 - \frac{{\ell_{0}}^{2}}{\ell^{2}}}^{- M} - 1} \sum_{\alpha \in \aset{K} \setminus \set{0}} \abs{\hat{f}_{\alpha}}.
\end{equation}
In other words, this theorem isolates the optimization error introduced by truncating the \glsdashedorshort{SOS} hierarchy at finite level $\ell$.

Now, consider the function $g \from \ccinterval{0}{\frac{1}{2}} \to \reals$ defined by
\begin{equation}
  g(x)
    = (1 - x)^{-M}
\end{equation}
Since $g$ is differentiable on $\oointerval{0}{\frac{1}{2}}$ and continuous on $\ccinterval{0}{\frac{1}{2}}$, by the mean value theorem, for every $x \in \oointerval{0}{\frac{1}{2}}$, there exists some $\xi \in \oointerval{0}{x}$ such that
\begin{equation}
  (1 - x)^{-M} - 1
    = g(x) - g(0)
    = g'(\xi) x
    = M (1 - \xi)^{- (M + 1)} x
    \leq M 2^{M + 1} x.
\end{equation}
Thus, by setting $x = \ell_{0}^{2} / \ell^{2}$, we obtain that, for every $\ell \geq \sqrt{2} \ell_{0}$,
\begin{equation}
\label{eq:SOSErrorBound}
  \abs{\hat{f}^{\star} - \hat{f}_{\ell}^{\star}}
    \leq \frac{\ell_{0}^{2}}{\ell^{2}} M 2^{M + 1} \sum_{\alpha \in \aset{K} \setminus \set{0}} \abs{\hat{f}_{\alpha}}.
\end{equation}
Hence the \glsdashedorshort{SOS} error decays like $\bigoh(\ell^{-2})$ once $\ell$ is beyond the threshold set by $\ell_0$.

Moreover, observe that, by the definition of $\hat{f}_{\alpha}$, we have that
\begin{subequations}
\begin{align}
  \sum_{\alpha \in \aset{K} \setminus \set{0}} \abs{\hat{f}_{\alpha}}
    &\leq \sum_{\alpha \in \aset{K}} \xpar[\big]{\abs{\hat{f}_{\alpha} - f_{\alpha}} + \abs{f_{\alpha}}} \\
    &\leq \frac{\epsilon}{2} + \sum_{\alpha \in \aset{K}} \abs{f_{\alpha}} \\
    &= \frac{\epsilon}{2} + \sum_{\alpha \in \aset{K} } \frac{1}{\abs{\aset{S}}} \abs[\bigg]{\sum_{\phi \in \aset{S}} f(\phi) \exp\xpar[\big]{-\im \inner{\alpha}{\phi}}} \\
    &\leq \frac{\epsilon}{2} + \frac{\abs{\aset{K}}}{\abs{\aset{S}}} \sum_{\phi \in \aset{S}} \abs{f(\phi)} \\
    &\leq \frac{\epsilon}{2} + \abs{\aset{K}} \max_{\phi \in \aset{S}} \abs{f(\phi)} \\
    &\leq \frac{\epsilon}{2} + \abs{\aset{K}} \max_{\phi \in \aset{S}} \abs[\Big]{\braket[\big]{0}{\amatrix{U}^{\A}(\phi) \amatrix{O} \amatrix{U}(\phi)}{0}}.
\end{align}
\end{subequations}
In addition, since $\amatrix{U}(\theta)$ is unitary for all $\theta \in \reals^{M}$, i.e., its spectrum lies in the unit circle, by the Cauchy--Schwarz inequality, we have that
\begin{equation}
  \abs[\Big]{\braket[\big]{0}{\amatrix{U}^{\A}(\phi) \amatrix{O} \amatrix{U}(\phi)}{0}}
    \leq \norm{\amatrix{O}}
    \leq O_{\max},
    \qquad \forall \phi \in \aset{S}.
\end{equation}
Thus, we have that
\begin{equation}
\label{eq:CoefficientSumBound}
  \sum_{\alpha \in \aset{K} \setminus \set{0}} \abs{\hat{f}_{\alpha}}
    \leq \frac{\epsilon}{2} + \abs{\aset{K}}\, O_{\max}.
\end{equation}
This coefficient-sum bound turns the abstract SOS error bound in \citet{bach_2023_exponential}  into an explicit bound computable in polynomial time in $n$ and $1 / \epsilon$.

Concretely, we require that the right-hand side in \eqref{eq:SOSErrorBound} is at most $\epsilon / 2$.
Thus, rearranging terms, and assuming $\ell \geq \sqrt{2}\,\ell_0$ (the condition used to obtain $\bigoh(\ell^{-2})$ decay), it suffices to choose $\ell$ such that
\begin{equation}
  \ell
    \geq \frac{2 \ell_{0}}{\sqrt{\epsilon}} \sqrt{M 2^{M} \sum_{\alpha \in \aset{K} \setminus \set{0}} \abs{\hat{f}_{\alpha}}}.
\end{equation}
Furthermore, by \eqref{eq:CoefficientSumBound}, we have that
\begin{equation}
  \frac{2 \ell_{0} }{\sqrt{\epsilon}} \sqrt{M 2^{M} \sum_{\alpha \in \aset{K} \setminus \set{0}} \abs{\hat{f}_{\alpha}}} 
    \leq \frac{2 \ell_{0}}{\sqrt{\epsilon}} \sqrt{M 2^{M} \xpar[\Big]{\frac{\epsilon}{2} + \abs{\aset{K}}\, O_{\max}}}
    \leq \frac{2 \ell_{0}}{\sqrt{\epsilon}} \sqrt{M 2^{M} \xpar[\Big]{\frac{1}{2} + \abs{\aset{K}}\, O_{\max}}}.
\end{equation}
Thus, it suffices to choose the relaxation order $\ell$ such that
\begin{subequations}
\label{eq:RequiredSOSLevel}
\begin{align}
  \ell
    &\geq \max\set[\bigg]{\sqrt{2}\, \ell_{0}, \frac{2 \ell_{0}}{\sqrt{\epsilon}} \sqrt{M 2^{M} \xpar[\Big]{\frac{1}{2} + \abs{\aset{K}}\, O_{\max}}}} \\
    &= \max\set[\bigg]{\sqrt{2}\, \ell_{0},\frac{2 \ell_{0}}{\sqrt{\epsilon}} \sqrt{M 2^{M} \xpar[\Big]{\frac{1}{2} + (2K\, \Delta_{\max} + 1)^{M}\, O_{\max}}}}
\end{align}
\end{subequations}
Moreover, since $M$ is constant with respect to $n$, and $K$, $\Delta_{\max}$, and $O_{\max}$ are polynomial in $n$, we have that the required order $\ell$ of the \glsdashedorshort{SOS} relaxation is at most polynomial in $n$, and $1 / \epsilon$.
Furthermore, since the complexity of solving the \glsdashedorshort{SOS} relaxation at order $\ell$ is $\bigoh\xpar[\big]{(2\ell + 1)^{M}}$ and $M$ is constant with respect to $n$, we conclude that the \glsdashedorshort{SOS} relaxation at order $\ell$ can be solved in polynomial time in $n$, and $1 / \epsilon$.

\paragraph{Step 5: Computing the final guarantee and complexity.} Combining the conditioned probability event with the two $\epsilon / 2$ error contributions gives the final \param{(\epsilon, \delta)}-guarantee.
In particular, by the triangle inequality, we have that at relaxation order $\ell$ satisfying \eqref{eq:RequiredSOSLevel},
\begin{equation}
  \abs{f^{\star} - \hat{f}_{\ell}^{\star}}
    \leq \abs[\big]{f^{\star} - \hat{f}^{\star}} + \abs[\big]{\hat{f}^{\star} - \hat{f}_{\ell}^{\star}}
    \leq \frac{\epsilon}{2} + \frac{\epsilon}{2}
    = \epsilon.
\end{equation}
Recalling that we conditioned on the event in \eqref{eq:UnionBound}, this bound holds with probability at least $1 - \delta$, so $\hat{f}_{\ell}^{\star}$ satisfies the desired accuracy guarantee. \qed{}

\subsection{Extraction of optimal \glsfmtshort{PQC} parameters from the \glsfmtshort{SOS} solution}
\label{sec:PQCOptimizerExtraction}

The proof of \cref{thm:EfficientApproximationAlgorithm} also provides insights on how to extract an approximate optimizer of the objective function $f$ from the solution of the \glsdashedorshort{SOS} relaxation at order $\ell$.
This is particularly relevant in practice, where one is often interested not only in approximating the optimal value $f^{\star}$, but also in finding parameters $\theta$ that (approximately) minimize the \gls{PQC} optimization problem.
The key enablers are the uniform approximation error bound in \eqref{eq:ObjectiveFunctionErrorBound}, which holds for the approximation $\hat{f}$ of $f$ constructed from \cref{alg:EfficientApproximationAlgorithm}, and the Flat Extension Theorem for the max-degree trigonometric moment hierarchy, which 
we introduce in \cref{thm:FlatExtension}.

As we discussed in \cref{sec:Preliminaries}, a Flat Extension Theorem for the truncated trigonometric moment hierarchy guarantees that, if the dual moment sequence of the relaxation at order $\ell + 1$ admits a flat extension, then a global minimizer of $\hat{f}$ can be extracted from the solution of the relaxation.
We now state the main results of this \lcnamecref{sec:PQCOptimizerExtraction}.

\begin{theorem}[Optimal \glsfmtshort{PQC} parameter extraction]
\label{thm:OptimizerExtraction}
  Consider an instance of the \gls{PQC} optimization problem over an \param{n}-qubit system with objective function $f \from \reals^{M} \to \reals$.
  Fix some error $\epsilon > 0$, and let $\hat{f} \from \reals^{M} \to \reals$ be \pgls{HTP} such that
  \begin{equation}
    \abs[\big]{f(\theta) - \hat{f}(\theta)} 
      \leq \frac{\epsilon}{2},
    \qquad \forall \theta \in \ccinterval{-\pi}{\pi}^{M}.
  \end{equation}
  Let $\amatrix{M_{\ell  + 1}}(y)$ denote the moment matrix of the order\nobreakdash-$(\ell + 1)$ max-degree trigonometric moment hierarchy for the \gls{HTP} optimization problem
  \begin{equation}
    \min_{\theta \in \ccinterval{-\pi}{\pi}^{M}} \hat{f}(\theta).
  \end{equation}
  If the moment sequence $y$ at order $\ell + 1$ of the hierarchy admits a flat extension, i.e., $\rank \amatrix{M_{\ell + 1}}(y) = \rank \amatrix{M_{\ell}}(y)$, where $\amatrix{M_{\ell}}(y)$ denotes the truncation of $\amatrix{M_{\ell + 1}}(y)$ to max-degree $\ell$, then a global minimizer $\hat{\theta}^{\star} \in \ccinterval{-\pi}{\pi}^{M}$ of $\hat{f}$ can be extracted from $y$ in time polynomial in $\ell$ (for fixed $M$) such that
  \begin{equation}
    \abs[\big]{f(\hat{\theta}^{\star}) - f^{\star}} 
      \leq \epsilon.
  \end{equation}
\end{theorem}

The proof of the above \lcnamecref{thm:OptimizerExtraction} relies on the following \emph{Flat Extension Theorem} for the max-degree trigonometric moment hierarchy.
We remark that extraction guarantees for this hierarchy are not available in the literature (recall that \citet{josz_2018_lasserre} analyze the total-degree hierarchy), therefore we need to derive them independently. 
Specifically, similarly to \citet{josz_2018_lasserre}, and following the template set by \citet{curto_2000_truncated,laurent_2009_generalized}, we obtain the following \lcnamecref{thm:FlatExtension}.

\begin{theorem}[Flat Extension Theorem]
\label{thm:FlatExtension}
  Suppose we solve level $\ell + 1$ of the max-degree trigonometric moment hierarchy and obtain an optimal sequence $y = (y_{\alpha, \beta})_{\norm{\alpha}_{\infty}, \norm{\beta}_{\infty} \leq \ell + 1}$ that moreover satisfies the flat extension condition
  \begin{equation}
  \label{eq:FlatExtensionCondition}
    \rank \amatrix{M_{\ell + 1}}(y) 
      = \rank \amatrix{M_{\ell}}(y),
  \end{equation}
  where $\amatrix{M_{\ell}}(y)$ denotes the truncation of $\amatrix{M_{\ell + 1}}(y)$ to max-degree $\ell$.
  Then there exists a positive atomic representing measure for $y$, supported on the \param{M}-torus, with at most $\rank \amatrix{M_{\ell}}(y)$ atoms, and each atom corresponds to a global minimizer of the \gls{HTP} optimization problem.
\end{theorem}

Since the proof of \cref{thm:FlatExtension} is rather technical, we first provide a sketch of the main ideas, in order to provide the proper intuition to understand the proof of \cref{thm:OptimizerExtraction}.
We defer the formal proof of \cref{thm:FlatExtension} to the following \lcnamecref{sec:FlatExtensionProof}.

\paragraph{Proof sketch of \cref{thm:FlatExtension}.} Let us start with an optimal order\nobreakdash-$(\ell + 1)$ moment sequence $y$ that satisfies the flat extension condition in \eqref{eq:FlatExtensionCondition}.
Because $y$ is feasible for the moment relaxation, it satisfies the constraints in the definition of the trigonometric moment hierarchy (cf.~\cref{eq:TrigonometricMomentRelaxation}); that is:
\begin{itemize}
  \item The sequence is normalized, i.e., $y_{0,0} = 1$.
  \item The moment matrix is \gls{PSD}.
  \item The moment matrix is Toeplitz.
\end{itemize}
We rely on those properties throughout the proof.

\paragraph{Step 1: Constructing the shift operators.} The main objects of interest are the shift operators $\amatrix{T}_{1}$, \dots, $\amatrix{T}_{M} \from \complex^{R} \to \complex^{R}$, where $R$ is the rank of the moment matrix.
These are linear operators defined to shift the monomial indices of the moment matrix $\amatrix{M_{\ell + 1}}(y)$.
The \gls{PSD} and Toeplitz properties of the moment matrix ensure that these operators are well-defined.

\paragraph{Step 2: Showing that the shift operators admit a common eigenbasis.} Ultimately, we are going to construct a positive atomic representing measure for $y$ from the spectral decomposition of the shift operators, but first we need to show that $\amatrix{T}_{1}$, \dots, $\amatrix{T}_{M}$ are \emph{simultaneously diagonalizable}, i.e., they have a common eigenbasis.

Indeed, by the Toeplitz property of the moment matrix, we can show that 
\begin{equation}
 \amatrix{T}_{i} \amatrix{T}_{j} 
    = \amatrix{T}_{j} \amatrix{T}_{i},
    \qquad \forall i, j \in \range{M},
\end{equation}
i.e., the shift operators pairwise commute.
Moreover, each $\amatrix{T}_{j}$ preserves inner products, and therefore, it is unitary.
Taken together, these two properties imply that $\amatrix{T}_{1}$, \dots, $\amatrix{T}_{M}$ are normal, and thus, simultaneously diagonalizable.

\paragraph{Step 3: Building an atomic representing measure.} Letting $\amatrix{T}_{j} = \amatrix{P} \amatrix{D}_{j} \amatrix{P}^{*}$ be the spectral decomposition of $\amatrix{T}_{j}$, where $\amatrix{P}$ is unitary and $\amatrix{D}_{j}$ is diagonal unitary, we can show that the convex combination of Dirac measures defined as
\begin{equation}
  \mu
    \eqdef \sum_{k = 1}^{R} \abs{\inner{z_{0}}{\amatrix{P}_{\colon\!, k}}}^{2} \delta_{d^{(k)}}
\end{equation}
is a positive atomic representing measure for $y$.
Here, $\amatrix{P}_{\colon\!, k}$ denotes the \param{k}-th column of $\amatrix{P}$, i.e., the \param{k}-th eigenvector of $\amatrix{T}_{1}$, \dots, $\amatrix{T}_{M}$, while $d^{(k)} \in \complex^{M}$ is an atom defined from the diagonal entries of $\amatrix{D}_{1}$, \dots, $\amatrix{D}_{M}$; that is,
\begin{equation}
  d^{(k)}_{j} 
    \eqdef \conj{(\amatrix{D}_{j})_{k, k}},
    \qquad \forall j \in \range{M},
\end{equation}
the fact that $\mu$ is a convex combination follows from the normalization of $y$.

\paragraph{Step 4: Certifying global optimality of the atoms.} It remains to show that each atom $d^{(k)}$ is a global minimizer of the \gls{HTP} optimization problem.
This is a consequence of the fact that each $d^{(k)}$ lies on the \param{M}-torus.
Indeed, since each $\amatrix{D}_{j}$ is diagonal unitary, we have that, for each $j \in \range{M}$ and $k \in \range{R}$,
\begin{equation}
  \abs{d^{(k)}_{j}}^{2}
    = \abs[\big]{\conj{(\amatrix{D}_{j})_{k, k}}}^{2}
    = \abs[\big]{(\amatrix{D}_{j})_{k, k}}^{2}
    = 1.
\end{equation}
Then optimality of the atoms follows from a standard convexity argument, and with the Flat Extension Theorem and the necessary notation in place, we can now prove \cref{thm:OptimizerExtraction}.

\paragraph{Proof of \cref{thm:OptimizerExtraction}.} By \cref{thm:FlatExtension}, if the moment sequence $y$ at order $\ell + 1$ of the hierarchy admits a flat extension, then there exists a positive atomic representing measure for $y$, supported on the \param{M}-torus, with at most $R = \rank \amatrix{M_{\ell + 1}}(y)$ atoms, and each atom corresponds to a global minimizer of the \gls{HTP} optimization problem.
In particular, these atoms can be extracted from the shift operators $\amatrix{T}_{1}, \dots, \amatrix{T}_{M}$, which can be computed from $\amatrix{M_{\ell + 1}}(y)$ in polynomial time in $\ell$.

To see this, for each $j \in \range{M}$, take a full-rank principal submatrix $\amatrix{M}$ of $\amatrix{M_{\ell}}(y)$ of size $R \times R$, indexed by some set $\aset{A}$.
Let $\amatrix{M}^{(j)}$ denote the $R \times R$ submatrix of $\amatrix{M_{\ell + 1}}(y)$ with row index set $\aset{A}$ and column index set $\aset{A} + \unitvector_{j}$.
Then
\begin{equation}
  \hat{\amatrix{T}}_{j}
    \eqdef \amatrix{M}^{\inv}\amatrix{M}^{(j)}
\end{equation}
is similar to the shift operator $\amatrix{T}_{j}$.
Now construct
\begin{equation}
  \amatrix{H}
    \eqdef \sum_{j = 1}^{M} \xpar[\big]{c_{j}\, \hat{\amatrix{T}}_{j} + \conj{c_{j}}\, \hat{\amatrix{T}}_{j}^{\A}},
\end{equation}
for a generic choice of $c \in \complex^{M}$, and compute an eigendecomposition
\begin{equation}
  \amatrix{H}
    = \amatrix{U}\amatrix{\Lambda}\amatrix{U}^{\A}.
\end{equation}
For generic $c$ (or deterministically for some $c$ in a constructed set of size $R (R - 1) / 2$; see, e.g., \citet{corless_1997_reordered}), the matrix $\amatrix{H}$ has simple spectrum, hence its eigenvectors recover the common eigenbasis of the commuting matrices $\hat{\amatrix{T}}_{1}, \dots, \hat{\amatrix{T}}_{M}$, up to permutation and phase, which are immaterial.
Therefore, for each $k \in \range{R}$, the \param{j}-th coordinate of the atom $d^{(k)}$ is recovered as
\begin{equation}
  d_{j}^{(k)}
    = \frac{\amatrix{U}_{\colon\!, k}^{\A}\hat{\amatrix{T}}_{j}\amatrix{U}_{\colon\!, k}}
          {\norm{\amatrix{U}_{\colon\!, k}}^{2}},
\end{equation}
where $\amatrix{U}_{\colon\!, k}$ denotes the $k$-th column of $\amatrix{U}$.
Thus, atom extraction is polynomial-time in the relaxation size, with dominant cost $\bigoh(R^{3})$ from the eigendecomposition of $\amatrix{H}$.
Since $R$ is dominated by the size of the index set of $\amatrix{M_{\ell}}(y)$, which is $(2\ell + 1)^{M}$, and $M$ is constant, we conclude that atom extraction can be done in polynomial time in $\ell$.

For each $k \in \range{R}$, let $\hat{\theta}^{(k)} \in \ccinterval{-\pi}{\pi}^{M}$ be the angular coordinates of the atom $d^{(k)}$, i.e., $d_{j}^{(k)} = \exp(\im \hat{\theta}_{j}^{(k)})$.
Then, since each atom corresponds to a global minimizer of the \gls{HTP} optimization problem, we have that $\hat{f}(\hat{\theta}^{(k)}) = \hat{f}^{\star}$ for all $k \in \range{R}$.
Moreover, since, by assumption, $\abs[\big]{\hat{f}(\theta) - f(\theta)} \leq \epsilon / 2$ for $\theta \in \ccinterval{-\pi}{\pi}^{M}$, it follows that $\abs[\big]{\hat{f}^{\star} - f^{\star}} \leq \epsilon / 2$ for all $k \in \range{R}$ (cf.~\cref{eq:GlobalOptimumErrorBound}).
Thus, by the triangle inequality, we obtain
\begin{equation}
  \abs[\big]{f(\hat{\theta}^{(k)}) - f^{\star}}
    \leq \abs[\big]{f(\hat{\theta}^{(k)}) - \hat{f}(\hat{\theta}^{(k)})} + \abs[\big]{\hat{f}(\hat{\theta}^{(k)}) - f^{\star}}
    \leq \frac{\epsilon}{2} + \frac{\epsilon}{2} 
    = \epsilon
    \qquad \forall k \in \range{R}.
\end{equation}
\qed{}

\Cref{thm:OptimizerExtraction} provides a way to extract an \param{\epsilon}-optimizer of $f$ under the standard \emph{flat extension} assumption for the dual moment sequence at order $\ell$.
It also implies a practical optimality certificate for the solution at order $\ell$ of the max-degree trigonometric moment/\glsdashedorshort{SOS} hierarchy; that is, if the dual moment sequence at order $\ell$ admits a flat extension, then the solution at order $\ell$ is optimal and an (approximate) global minimizer of $f$ can be extracted from the optimal moment sequence.
Indeed, checking whether the dual moment sequence is flat can be done efficiently by verifying \emph{a rank condition on the moment matrix}, which requires only polynomial time in $\ell$.

Finally, we emphasize that \cref{thm:OptimizerExtraction} assumes \pgls{HTP} surrogate $\hat{f}$ satisfying the uniform approximation bound with respect to $f$.
When $\hat{f}$ is constructed via \cref{alg:EfficientApproximationAlgorithm}, this guarantee follows from \cref{thm:EfficientApproximationAlgorithm} under \cref{ass:IntegerSpectralDifferences,ass:PQCComplexity,ass:ObservableComplexity}, on the high-probability event in \eqref{eq:UnionBound}, which occurs with probability at least $1 - \delta$.
Accordingly, the end-to-end procedure is polynomial in $n$, $1 / \epsilon$, and $\log(1 / \delta)$ for constructing $\hat{f}$ using the \gls{FFT} and solving the hierarchy, while the atom-extraction step is polynomial in $\ell$ (for fixed $M$).
Moreover, the number of queries to quantum hardware also mirrors the one required for constructing $\hat{f}$, i.e., it is polynomial in $n$, $1 / \epsilon$, and $\log(1 / \delta)$.

\subsubsection{Proof of the Flat Extension Theorem}
\label{sec:FlatExtensionProof}

Consider an arbitrary \gls{HTP} $f$, and recall we may identify $f$ either with its angle formulation $f(\theta)$ or with its complex formulation $f(z, \conj{z})$ via the reparametrization $z_{j} = \exp(\im \theta_{j})$, for all $j \in \range{M}$, so $f(z, \conj{z}) = f(\exp(\im \theta), \exp(-\im \theta))$ (cf.~\cref{sec:HermitianPolynomialOptimization}).
Since in this proof we rely heavily on the definition of the order\nobreakdash-$(\ell + 1)$ trigonometric moment relaxation in \eqref{eq:TrigonometricMomentRelaxation}, which is more naturally expressed in terms of the complex formulation, we adopt this formulation throughout.

Moreover, recall that at a feasible solution of the order\nobreakdash-$(\ell + 1)$ relaxation, the moment sequence $y$ satisfies the normalization $y_{0, 0} = 1$, and the moment $\amatrix{M_{\ell + 1}}(y)$ is a Toeplitz \glsdashedorshort{PSD} matrix, i.e., it satisfies the Toeplitz constraints
\begin{equation}
\label{eq:ToeplitzConstraints}
  y_{\alpha + \unitvector_{j}, \beta + \unitvector_{j}}
    = y_{\alpha, \beta}
    \qquad \forall \alpha, \beta \in \aset{A}^{(j)}_{\ell + 1},
\end{equation}
We are going to rely on those properties throughout the proof.

We now split the argument into four steps.
In particular, we begin by defining the shift operators taking advantage of the above Toeplitz constraints, then we show they are commuting unitaries, and therefore they have a common eigenbasis.
Finally, we build a finite atomic representing measure from their joint spectrum, and certify that every support atom is globally optimal.

\paragraph{Step 1: Constructing the shift operators.} Throughout this proof we define the following index sets:
\begin{subequations}
\begin{align}
  \aset{A}_{\ell}
    &\eqdef \set[\big]{\alpha \in \naturals^{M} \given \alpha_{j} \leq \ell,\ \forall j \in \range{M}} \\
  \aset{A}^{(j)}_{\ell + 1}
    &\eqdef \set[\big]{\alpha \in \naturals^{M} \given \alpha_{j} \leq \ell \ \text{and},\ \alpha_{i} \leq \ell + 1,\ \forall i \in \range{M}\ \text{with $i \neq j$}},
\end{align}
\end{subequations}
for which we have the inclusions
\begin{equation}
  \aset{A}_{\ell}
    \subseteq \aset{A}^{(j)}_{\ell + 1} 
    \subseteq \aset{A}_{\ell + 1}, 
    \qquad \forall j \in \range{M}.
\end{equation}
Now, set $R \eqdef \rank \amatrix{M_{\ell + 1}}(y) = \rank \amatrix{M_{\ell}}(y)$.
As $\amatrix{M_{\ell + 1}}(y)$ is \pglsdashedorshort{PSD} matrix, it admits a Gram factorization  
\begin{equation}
  y_{\alpha, \beta}
    = z_{\alpha}^{\A} z_{\beta},
    \qquad \forall \alpha, \beta \in \aset{A}_{\ell + 1},
\end{equation}
where $z_{\alpha} \in \complex^{R}$ for all $\alpha \in \aset{A}_{\ell + 1}$.  
Then by the rank condition we have that 
\begin{equation}
\label{eq:Span}
  \lin \set{z_{\alpha} \given \alpha \in \aset{A}_{\ell}}
    \equiv \lin \set{z_{\alpha} \given \alpha \in \aset{A}_{\ell + 1}} 
    \equiv \complex^{R}.
\end{equation}
		
For $j \in \range{M}$, we define the \emph{shift operators} $\amatrix{T}_{j} \from \complex^{R} \to \complex^{R}$ by
\begin{equation}
  \amatrix{T}_{j} z_{\alpha} 
    \eqdef z_{\alpha + \unitvector_{j}},
    \qquad \forall \alpha \in \aset{A}^{(j)}_{\ell + 1},
\end{equation}
which we extend by linearity to any vector in $\complex^{R}$.
Indeed, this is possible because $\lin \set{z_{\alpha} \given \alpha \in \aset{A}_{\ell}} \equiv \complex^{R}$ (by \cref{eq:Span}), and $\aset{A}_{\ell} \subseteq \aset{A}^{(j)}_{\ell + 1}$. 
Intuitively, $\amatrix{T}_{j}$ increases the \param{j}-th index of a moment basis vector, i.e., it maps the Gram vector associated with $z^{\alpha}$ to the one associated with $z^{\alpha + \unitvector_{j}}$.
However, as $\set{z_{\alpha} \given \alpha \in \aset{A}^{(j)}_{\ell + 1}}$ are not necessarily linearly independent, we need to show that the shift operators are well defined. 

The only subtlety is that the family $\set{z_{\alpha}}$ may be linearly dependent, so we must verify that this definition is independent of the chosen representation of a vector in $\complex^{R}$.
For this, we show that $\sum_{\alpha \in \aset{A}^{(j)}_{\ell + 1}} \lambda_{\alpha} z_{\alpha} = 0$ implies that $\sum_{\alpha \in \aset{A}^{(j)}_{\ell + 1}} \lambda_{\alpha} \amatrix{T}_{j} z_{\alpha} = 0$. 
Specifically, we show that
\begin{equation}
  \norm[\Big]{\sum_{\alpha \in \aset{A}^{(j)}_{\ell + 1}} \lambda_{\alpha} z_{\alpha}}_{2}^{2}
    = \norm[\Big]{\sum_{\alpha \in \aset{A}^{(j)}_{\ell + 1}} \lambda_{\alpha} \amatrix{T}_{j} z_{\alpha}}_{2}^{2}
\end{equation}
for the Euclidean norm.    
Indeed, by the Toeplitz constraint in \eqref{eq:ToeplitzConstraints}, it follows that
\begin{subequations}
\begin{alignat}{2}
  \norm[\Big]{\sum_{\alpha \in \aset{A}^{(j)}_{\ell + 1}} \lambda_{\alpha} z_{\alpha}}_{2}^{2}
    &= \sum_{\alpha, \beta \in \aset{A}^{(j)}_{\ell + 1}} \conj{\lambda_{\alpha}} \lambda_{\beta} z^{\A}_{\alpha} z_{\beta} \\
    &= \sum_{\alpha, \beta \in \aset{A}^{(j)}_{\ell + 1}} \conj{\lambda_{\alpha}} \lambda_{\beta} y_{\alpha, \beta} \\
    &= \sum_{\alpha, \beta \in \aset{A}^{(j)}_{\ell + 1}} \conj{\lambda_{\alpha}} \lambda_{\beta} y_{\alpha + \unitvector_{j}, \beta + \unitvector_{j}} 
      &&\qquad \text{by \cref{eq:ToeplitzConstraints}}\\
    &= \sum_{\alpha, \beta \in \aset{A}^{(j)}_{\ell + 1}} \conj{\lambda_{\alpha}} \lambda_{\beta} z^{\A}_{\alpha + \unitvector_{j}} z_{\beta + \unitvector_{j}} \\
    &= \sum_{\alpha, \beta \in \aset{A}^{(j)}_{\ell + 1}} \conj{\lambda_{\alpha}} \lambda_{\beta} (\amatrix{T}_{j} z_{\alpha})^{\A} \amatrix{T}_{j} z_{\beta} \\
    &= \norm[\Big]{\sum_{\alpha \in \aset{A}^{(j)}_{\ell + 1}} \lambda_{\alpha} \amatrix{T}_{j} z_{\alpha}}_{2}^{2}.
\end{alignat}
\end{subequations}
Hence the shift preserves all linear relations among the generators, so $\amatrix{T}_{j}$ is well defined.

\paragraph{Step 2: Showing that the shift operators admit a common eigenbasis.}
Next, we show that the shift operators $\amatrix{T}_{1}, \ldots, \amatrix{T}_{M}$ are simultaneously diagonalizable, i.e., they have a common eigenbasis.
At a high level, Toeplitz invariance gives isometric shifts, and commutativity then lets us diagonalize all shifts in one common basis.

First, we show that the shift operators are unitary. 
In particular, we show that they are isometries, i.e., they preserve inner products, and therefore, they are unitary.
To see this, consider arbitrary $j \in \range{M}$. 
By the Toeplitz constraints in \cref{eq:ToeplitzConstraints}, and since $\aset{A}_{\ell} \equiv \bigcap_{r = 1}^{M} \aset{A}^{(r)}_{\ell + 1}$, we have that
\begin{equation}
  y_{\alpha + \unitvector_{j}, \beta + \unitvector_{j}}
    = y_{\alpha, \beta}
    \qquad \forall \alpha, \beta \in \aset{A}_{\ell},
\end{equation}
From this it follows, for all $ \alpha, \beta \in \aset{A}_{\ell}$, that
\begin{equation}
  \inner{\amatrix{T}_{j} z_{\alpha}}{\amatrix{T}_{j} z_{\beta}}
    = \inner{z_{\alpha + \unitvector_{j}}}{z_{\beta + \unitvector_{j}}}
    = y_{\alpha + \unitvector_{j}, \beta + \unitvector_{j}}
    = y_{\alpha, \beta}
    = \inner{z_\alpha}{z_\beta},
\end{equation}
and since, by \eqref{eq:Span}, we also have $\lin \set{z_{\alpha} \given \alpha \in \aset{A}_{\ell}} = \complex^{R}$, we conclude that
\begin{equation}
  \inner{\amatrix{T}_{j} u}{\amatrix{T}_{j} v} 
    = \inner{u}{v} 
    \qquad \qquad \forall u, v \in \complex^{R},
\end{equation}
and therefore $\amatrix{T}_{j}$ is an isometry, and thus unitary.

Next, we show that $\amatrix{T}_{1}, \ldots, \amatrix{T}_{M}$ pairwise commute, i.e., for all $i$, $j \in \range{M}$, and $\lambda \in \complex^{R}$ we have $\amatrix{T}_{i} \amatrix{T}_{j} \lambda = \amatrix{T}_{j} \amatrix{T}_{i} \lambda$. 
This is clearly true for $i = j$. 
While, for $i \neq j$, we argue that, since every $\lambda \in \complex^{R}$ has decomposition $\sum_{\alpha \in \aset{A}_{\ell}} \lambda_{\alpha} z_{\alpha}$, we can write
\begin{equation}
  \amatrix{T}_{i} \amatrix{T}_{j} \lambda
  = \amatrix{T}_{i} \amatrix{T}_{j} \xpar[\bigg]{\sum_{\alpha \in \aset{A}_{\ell}} \lambda_{\alpha} z_{\alpha}}
  = \sum_{\alpha \in \aset{A}_{\ell}} \lambda_{\alpha} \amatrix{T}_{i} \amatrix{T}_{j} z_{\alpha}
  = \sum_{\alpha \in \aset{A}_{\ell}} \lambda_{\alpha} \amatrix{T}_{j} \amatrix{T}_{i} z_{\alpha}
  = \amatrix{T}_{j} \amatrix{T}_{i} \xpar[\bigg]{\sum_{\alpha \in \aset{A}_{\ell}} \lambda_{\alpha} z_{\alpha}}
  = \amatrix{T}_{j} \amatrix{T}_{i} \lambda,
\end{equation}
where the third equality is not immediate and it remains to show that for all $i \neq j$ we have 
\begin{equation}
  \amatrix{T}_{i} \amatrix{T}_{j} z_{\alpha} 
    = \amatrix{T}_{j} \amatrix{T}_{i} z_{\alpha}
    \qquad \forall \alpha \in \aset{A}_{\ell}.
\end{equation}

To establish this fact, note that $\aset{A}_{\ell} \equiv \bigcap_{m = 1}^{M} \aset{A}^{(m)}_{\ell + 1}$, implying that
\begin{equation}
  \amatrix{T}_{i} \amatrix{T}_{j} z_{\alpha} 
  = \amatrix{T}_{i} z_{\alpha + \unitvector_{j}} 
  = z_{\alpha+\unitvector_{j} + \unitvector_{i}} 
  = z_{\alpha+\unitvector_{i} + \unitvector_{j}} 
  = \amatrix{T}_{j} z_{\alpha + \unitvector_{i}}
  = \amatrix{T}_{j} \amatrix{T}_{i} z_{\alpha}
  \qquad \forall \alpha \in \aset{A}_{\ell}.
\end{equation}
Indeed, the first and fifth equalities follow as $\alpha \in \aset{A}^{(j)}_{\ell + 1} \cap \aset{A}^{(i)}_{\ell + 1}$, the second equality as $\alpha + \unitvector_{j} \in \aset{A}^{(i)}_{\ell + 1}$, and the fourth equality as $\alpha \in \aset{A}^{(i)}_{\ell + 1}$.
Summarizing, we have shown that the shift operators are unitary and pairwise commute.
These properties imply that the operators are \emph{normal} and therefore the spectral theorem gives one common eigenbasis for all of them.
        
\paragraph{Step 3: Building an atomic representing measure.} Since the shift operators are simultaneously diagonalizable, there exists a unitary $\amatrix{P} \in \complex^{R \times R}$, and diagonal \emph{unitary} matrices $\amatrix{D}_{1}$, \dots, $\amatrix{D}_{M} \in \complex^{R \times R}$, such that $\amatrix{T}_{j} = \amatrix{P} \amatrix{D}_{j} \amatrix{P}^{\A}$ for all $j \in \range{M}$.
Let $\amatrix{P}_{\colon\!, k}$, for $k \in \range{R}$ denote the columns of $\amatrix{P}$. 

Now, define the operators $\amatrix{T}^{\alpha} \from \complex^{R} \to \complex^{R}$ for $\alpha \in \aset{A}_{\ell + 1}$ by
\begin{equation}
  \amatrix{T}^{\alpha}
    \eqdef \amatrix{T}_{1}^{\alpha_{1}} \cdots \amatrix{T}_{M}^{\alpha_{M}},
\end{equation}
which are well defined as the shift operators pairwise commute.
Then, by the definition of the shift operators $\amatrix{T}_{1}, \ldots, \amatrix{T}_{M}$, we have that
\begin{equation}
  \amatrix{T}^{\alpha} z_{0}
    = z_{\alpha},
    \qquad \forall \alpha \in \aset{A}_{\ell + 1}.
\end{equation}
In other words, $\amatrix{T}^{\alpha}$ shifts the vector $z_{0}$ to the vector $z_{\alpha}$ by applying the shift operators $\amatrix{T}_{1}, \ldots, \amatrix{T}_{M}$ according to the multi-index $\alpha$.

Similarly, define the matrices $\amatrix{D}^{\alpha} \in \complex^{R \times R}$ for $\alpha \in \integers^{M}$ by
\begin{equation}
  \amatrix{D}^{\alpha} 
    \eqdef \amatrix{D}_{1}^{\alpha_{1}} \cdots \amatrix{D}_{M}^{\alpha_{M}},
\end{equation}
which are well defined as the diagonal matrices $\amatrix{D}_{1}, \ldots, \amatrix{D}_{M}$ pairwise commute.
Observe that:
\begin{enumerate}
  \item Since $\amatrix{D}_{1}$, \dots, $\amatrix{D}_{M}$ are diagonal matrices, it follows that $\amatrix{D}^{\alpha}$ is diagonal. 
  In particular, its \param{k}-th diagonal entry is $\prod_{j = 1}^{M} (\amatrix{D}_{j})_{k, k}^{\alpha_{j}}$.
  \item Since $\amatrix{D}_{1}$, \dots, $\amatrix{D}_{M}$ are unitary, it follows that $\amatrix{D}^{\alpha}$ is unitary, and therefore $\xpar[\big]{\amatrix{D}^{\alpha}}^{\A} = (\amatrix{D}^{-\alpha})$.
\end{enumerate}
Thus, by combining the above, we have
\begin{equation}
\label{eq:DiagonalUnitaryOperatorProperty}
  \xpar[\big]{\amatrix{D}^{\alpha}}^{\A} \amatrix{D}^{\beta} 
    = \amatrix{D}^{\beta - \alpha},
    \qquad \forall \alpha, \beta \in \integers^{M}.
\end{equation}
Finally, since $\amatrix{P}$ is also unitary, we have that
\begin{equation}
\label{eq:GeneralizedShiftOperatorDiagonalization}
  \amatrix{T}^{\alpha} 
    = \amatrix{P} \amatrix{D}_{1}^{\alpha_{1}} \cdots \amatrix{D}_{M}^{\alpha_{M}} \amatrix{P}^{\A}
     = \amatrix{P} \amatrix{D}^{\alpha} \amatrix{P}^{\A},
    \qquad \forall \alpha \in \aset{A}_{\ell + 1}.
\end{equation}

Then, for each $\alpha$, $\beta \in \aset{A}_{\ell + 1}$, it follows that
\begin{subequations}
\begin{alignat}{2}
  y_{\alpha, \beta}
    &= z_{\alpha}^{\A} z_{\beta} \\
    &= \xpar[\big]{\amatrix{T}^{\alpha} z_{0}}^{\A} \xpar[\big]{\amatrix{T}^{\beta} z_{0}} \\
    &= z_{0}^{\A} \xpar[\big]{\amatrix{T}^{\alpha}}^{\A} \amatrix{T}^{\beta} z_{0} \\
    &= z_{0}^{\A} \amatrix{P} \xpar[\big]{\amatrix{D}^{\alpha}}^{\A} \amatrix{D}^{\beta} \amatrix{P}^{\A} z_{0}
      &&\qquad \text{by \cref{eq:GeneralizedShiftOperatorDiagonalization}} \\
    &= z_{0}^{\A} \amatrix{P} \amatrix{D}^{\beta - \alpha} \amatrix{P}^{\A} z_{0}
      &&\qquad \text{by \cref{eq:DiagonalUnitaryOperatorProperty}} \\
    &= \sum_{k = 1}^{R} z_{0}^{\A} \amatrix{P}_{\colon\!, k} \amatrix{P}_{\colon\!, k}^{\A} z_{0}\, \amatrix{D}^{\beta - \alpha}_{k, k} \\
    &= \sum_{k = 1}^{R} \abs[\big]{\inner{z_{0}}{\amatrix{P}_{\colon\!, k}}}^{2}\, \amatrix{D}^{\beta - \alpha}_{k, k}.
\end{alignat}
\end{subequations}

For each $k \in \range{R}$, define the vector $d^{(k)} \in \complex^{M}$ whose \param{j}-th entry is the conjugate of the \param{k}-th diagonal entry of $\amatrix{D}_{j}$, i.e., $d^{(k)}_{j} \eqdef \conj{(\amatrix{D}_{j})_{k, k}}$.
Then, define the convex combination of Dirac measures 
\begin{equation}
  \mu
    \eqdef \sum_{k = 1}^{R} \abs{\inner{z_{0}}{\amatrix{P}_{\colon\!, k}}}^{2} \delta_{d^{(k)}}.
\end{equation}
This is indeed a convex combination as $\sum_{k = 1}^{R} \abs{\inner{z_{0}}{\amatrix{P}_{\colon\!, k}}}^{2} = \norm{z_{0}}_{2}^{2} = y_{0, 0} = 1$, which holds at a feasible solution of the trigonometric moment relaxation.
Therefore, the coefficients $\abs{\inner{z_{0}}{\amatrix{P}_{\colon\!, k}}}^{2}$ are nonnegative and sum to one, so $\mu$ is a probability measure supported on at most $R$ torus points.
In fact, $\mu$ is a representing measure for the truncated pseudo-moment sequence $(y_{\alpha, \beta})_{\alpha, \beta \in \aset{A}_{\ell + 1}}$, i.e., it is a measure that satisfies 
\begin{equation}
  y_{\alpha, \beta} 
    = \int_{\complex^{M}} z^{\alpha} \conj{z}^{\beta} d\,\mu 
    \qquad \forall \alpha, \beta \in \aset{A}_{\ell + 1}.
\end{equation}    
 
\paragraph{Step 4: Certifying global optimality of the atoms.} The last step is a convexity argument.
In particular, we show that the relaxation value $f^{\star}_{\ell + 1}$ is a convex combination of objective values at the extracted atoms $d^{(k)}$.
The argument concludes by showing that each atom $d^{(k)}$ in the support of $\mu$ corresponds to a global minimizer of the \gls{HTP} optimization problem.

First, observe that $\mu$ is supported on the complex \param{M}-torus $\complextorus^{M}$.
It suffices to show that $d^{(k)} \in \complextorus^{M}$ for each $k \in \range{R}$, i.e., $\abs{d^{(k)}_{j}} = 1$ for each $j \in \range{M}$.
Indeed, for each $k \in \range{R}$ and each $j \in \range{M}$, we have
\begin{equation}
  \abs{d^{(k)}_{j}}^{2}
    = \abs[\big]{\conj{(\amatrix{D}_{j})_{k, k}}}^{2}
    = \abs[\big]{(\amatrix{D}_{j})_{k, k}}^{2}
     = 1,
\end{equation}
because $\amatrix{D}_{j}$ is unitary.

It remains to show that each atom $d^{(k)}$ in the support of $\mu$ corresponds to a global minimizer of the \gls{HTP} optimization problem.%
For each $k \in \range{R}$, since $d^{(k)}$ is a feasible point of the \gls{HTP} optimization problem $\min_{z \in \complextorus^{M}} f(z, \conj{z})$, we have $f(d^{(k)}, \conj{d^{(k)}}) \geq f^{\star}$.
Furthermore, by the optimality of $y$ for the order\nobreakdash-$(\ell + 1)$ trigonometric moment relaxation, we have that
\begin{equation}
  f^{\star}_{\ell + 1}
    = \sum_{\alpha, \beta \in \aset{A}_{\ell + 1}} f_{\alpha, \beta} y_{\alpha, \beta}
    = \sum_{\alpha, \beta \in \aset{A}_{\ell + 1}} f_{\alpha, \beta} \int_{\complex^{M}} z^{\alpha} \conj{z}^{\beta} d\,\mu
    = \sum_{k = 1}^{R} \abs{\inner{z_{0}}{\amatrix{P}_{\colon\!, k}}}^{2} f(d^{(k)}, \conj{d^{(k)}}).
\end{equation}
Thus, since $f^{\star}_{\ell + 1} \leq f^{\star}$, it follows that $\sum_{k = 1}^{R} \abs{\inner{z_{0}}{\amatrix{P}_{\colon\!, k}}}^{2} f(d^{(k)}, \conj{d^{(k)}}) \leq f^{\star}$.
Therefore, because each atom is feasible and the convex combination is at most $f^{\star}$, every support atom must satisfy $f(d^{(k)}, \conj{d^{(k)}}) = f^{\star}$.
In other words, each atom $d^{(k)}$ in the support of $\mu$ corresponds to a global minimizer of the \gls{HTP} optimization problem. \qed{}

\section{Expressivity limitations of poly-depth constant-parameter \glsfmtshortpl{PQC}}
\label{se:ExpressivityLimitations}

We now address a couple of complexity considerations related to the existence of the additive \gls{FPRAS} in \cref{alg:EfficientApproximationAlgorithm} for the class of \gls{PQC} optimization problems satisfying \lcnamecrefs{ass:IntegerSpectralDifferences} in \cref{sec:Assumptions}.
In particular, let us address the following pressing questions:
\begin{enumerate}
  \item\label[question]{qst:Expressivity} How expressive is the class of \gls{PQC} optimization problems with a fixed parameter count?
  \item\label[question]{qst:Complexity} What does the existence of the \gls{FPRAS} imply, in conjunction with the hardness results of \citet{bittel_2021_training}, about the complexity of \gls{PQC} optimization problems?
\end{enumerate}

Our discussion began with the existence of quantum hardware that is capable of efficiently estimating the expectation value $f(\theta)$ of an observable $\amatrix{O}$ by performing measurements on a quantum state prepared by a parameterized quantum circuit $\amatrix{U}(\theta)$.
Nowhere in this discussion were we interested in representing $\amatrix{U}(\theta)$ or $\amatrix{O}$ explicitly.
Therefore, in order to discuss a problem instance $\aset{I}$ in a complexity-theoretic sense, we need to introduce standard assumptions on the representation of $\amatrix{U}(\theta)$ and $\amatrix{O}$ (\cref{ass:PQCPolynomialDescription,ass:ObservablePolynomialDescription}), and on the realization of the quantum query routine for $f(\theta)$, which in layman's terms corresponds to the existence of a classical procedure that, given $\amatrix{U}$, $\amatrix{O}$, and $\theta$, outputs parameterized quantum hardware capable of producing samples $X(\theta)$ with $\mean[\big]{X(\theta)} = f(\theta)$ (\cref{ass:QuantumQueryRoutineRealization}).
Formally, we have the following assumptions.

\begin{assumption}
\label{ass:PQCOPolynomialDescription}
  Consider an instance of the \gls{PQC} optimization problem over an \param{n}-qubit system.
  We assume that:
  \begin{enumerate}[ref=\theassumption.\arabic*]
    \item\label[assumption]{ass:PQCPolynomialDescription} The description of the \gls{PQC} $\amatrix{U}$ is $\bigtheta\xpar[\big]{\polytime(n)}$.
    \item\label[assumption]{ass:ObservablePolynomialDescription} The description of the observable $\amatrix{O}$ is $\bigtheta\xpar[\big]{\polytime(n)}$.
    \item\label[assumption]{ass:QuantumQueryRoutineRealization} Given $\theta \in \rationals^{M}$ with \emph{polynomial bit-length}, there exists a deterministic polynomial-time classical procedure that, given $(\amatrix{U}, \amatrix{O}, \theta)$, outputs a quantum query routine producing samples $X(\theta)$ with $\mean[\big]{X(\theta)} = f(\theta)$.
  \end{enumerate}
\end{assumption}

\noindent We can now formally define the class \gls{PQCO} problems satisfying our assumptions.

\paragraph{The \gls{PQCO} problem class.} Let \gls{PQCO} denote the class of \gls{PQC} optimization problems satisfying \cref{ass:IntegerSpectralDifferences,ass:PQCComplexity,ass:ObservableComplexity,ass:PQCOPolynomialDescription}.
By definition, \pgls{PQCO} instance $\aset{I}$ is identified by \pgls{PQC} $\amatrix{U}$ and an observable $\amatrix{O}$, i.e., $\aset{I} \equiv (\amatrix{U}, \amatrix{O})$.
Therefore, by \cref{ass:PQCPolynomialDescription,ass:ObservablePolynomialDescription}, the description of $\aset{I}$ is $\bigtheta\xpar[\big]{\polytime(n)}$, i.e., 
\begin{equation}
  \abs{\aproblem{I}} 
    = \bigtheta\xpar[\big]{\polytime(n)}.
\end{equation}
Furthermore, by \cref{ass:QuantumQueryRoutineRealization}, there exists a deterministic polynomial-time classical procedure that, given $(\aproblem{I}, \theta)$, outputs a quantum query routine producing samples $X_{\aproblem{I}}(\theta)$ with $\mean[\big]{X_{\aproblem{I}}(\theta)} = f_{\aproblem{I}}(\theta)$, where $f_{\aproblem{I}}$ denotes the objective function of $\aproblem{I}$.

To prove the expressiveness of the \gls{PQCO} problem class, we are going to analyze the induced class \gls{GAPPQCO} of promise problems (see, e.g., \citet{goldreich_2006_promise}) defined below.
This allows us to show the relationship between \gls{PQCO} class and the \gls{BPPBQP} class~\citep{bernstein_1997_quantum,bennett_1997_strengths}.
To make the discussion self-contained, let us now recall the definitions of the complexity classes \gls{BQP} and \gls{BPPBQP}.

\begin{definition}[The complexity classes \gls{BQP} and \gls{BPPBQP}]
\label{def:BPPBQP}
  \Gls{BQP} is the class of promise problems solvable by a polynomial-time quantum Turing machine with bounded error, and \gls{BPPBQP} is the class of promise problems solvable by a probabilistic polynomial-time classical oracle Turing machine with bounded error and access to \pgls{BQP} oracle.
\end{definition}

\noindent Clearly, $\gls{BQP} \subseteq \gls{BPPBQP}$, since \pgls{BPPBQP} machine can simulate \pgls{BQP} machine by making a single query to its \gls{BQP} oracle.
It's nontrivial, but also known, that $\gls{BPPBQP} \subseteq \gls{BQP}$.
That is, \pgls{BQP} machine can simulate \pgls{BPPBQP} machine with only polynomial overhead, and therefore $\gls{BPPBQP} \equiv \gls{BQP}$~\citep{bennett_1997_strengths}.

With these definitions in place, let us now define the class \gls{GAPPQCO} of promise problems induced by \gls{PQCO} optimization problems, and show that they are in \gls{BPPBQP} (and therefore in \gls{BQP}).
Later on, we exploit this relationship to establish the limitations to the expressivity of \gls{PQCO}.

\begin{definition}[The \gls{GAPPQCO} problem]
\label{def:GAPPQCO}
  Given a positive polynomial $p$, a $\gls{GAPPQCO}(p)$ problem is a promise problem with YES set
  \begin{equation}
    \set[\bigg]{
      (\aproblem{I}, a, b) 
      \given 
        f_{\aproblem{I}}^{\star} 
          \leq a,\ 
        b - a 
          \geq \frac{1}{p\xpar[\big]{\abs{\aproblem{I}}}}
    },
  \end{equation}
  and NO set
  \begin{equation}
    \set[\bigg]{
      (\aproblem{I}, a, b) 
      \given 
        f_{\aproblem{I}}^{\star} \geq b,\ 
        b - a \geq \frac{1}{p\xpar[\big]{\abs{\aproblem{I}}}}
    },
  \end{equation}
  where $a$, $b \in \rationals$, and $\aproblem{I} \in \gls{PQCO}$.
  Here, $f_{\aproblem{I}}^{\star}$ denotes the minimum value of $\aproblem{I}$.
  The condition 
  \begin{equation}
    b - a 
      \geq \frac{1}{p\xpar[\big]{\abs{\aproblem{I}}}}
  \end{equation}
  is known as the promise gap.
  When the polynomial $p$ is irrelevant, we simply write $\gls{GAPPQCO}$ to denote the class of problems.
\end{definition}

We are now ready to prove that \gls{GAPPQCO} is in \gls{BPPBQP} (for all polynomial $p$).
Essentially, this means that solving the promise problem associated with \gls{PQCO} optimization problems is no harder than solving a problem in \gls{BPPBQP}, and therefore in \gls{BQP}.
To show this, we use the existence of the \gls{FPRAS} in \cref{alg:EfficientApproximationAlgorithm} for the class of \gls{PQCO} optimization problems, and the guarantees of \cref{thm:EfficientApproximationAlgorithm}.

\begin{corollary}
\label{thm:GAPPQCOinBPPBQP}
  \Gls{GAPPQCO} is in \gls{BPPBQP} (and therefore in \gls{BQP}).
\end{corollary}
\begin{proof}
  Let $(\aproblem{I}, a, b)$ be an input satisfying the promise gap, i.e., $b - a \geq 1 / p\xpar[\big]{\abs{\aproblem{I}}}$, and set
  \begin{equation}
    \epsilon 
      \equiv \frac{b - a}{3}.
  \end{equation}
  Run the \gls{FPRAS} of \cref{alg:EfficientApproximationAlgorithm} on $\aproblem{I}$ with error parameter $\epsilon$ and failure probability $\delta = 1 / 3$, and let $\hat{f}_{\aproblem{I}}^{\star}$ be the returned estimate of $f_{\aproblem{I}}^{\star}$.

  \paragraph{Step 1: Bounding the runtime.} By \cref{thm:EfficientApproximationAlgorithm}, this algorithm runs in time polynomial in $n$, $1 / \epsilon$, $\log(1 / \delta)$, and $\abs{\aproblem{I}}$ (see the technical note below), and uses only a polynomial number of queries in $n$, $1 / \epsilon$, $\log(1 / \delta)$.
  Since
  \begin{equation}
    \frac{1}{\epsilon}
      = \frac{3}{b - a} 
      \leq 3p\xpar[\big]{\abs{\aproblem{I}}},
  \end{equation}
  it follows that $1 / \epsilon$ is polynomial in $\abs{\aproblem{I}}$. 
  Moreover, since $\abs{\aproblem{I}} = \bigtheta\xpar[\big]{\polytime(n)}$ (\cref{ass:PQCPolynomialDescription,ass:ObservablePolynomialDescription}), the overall runtime and the number of queries are polynomial in $\abs{\aproblem{I}}$.
  Hence, the above procedure is a probabilistic polynomial-time classical oracle Turing machine with bounded error and access to \pgls{BQP} oracle, which is exactly \pgls{BPPBQP} machine.

  \paragraph{Step 2: Defining the output.} To complete the proof, it remains to define the output of the machine.
  In particular, output YES if and only if
  \begin{equation}
    \hat{f}_{\aproblem{I}}^{\star}
      \leq \frac{a + b}{2};
  \end{equation}
  and output NO otherwise.
  If $f_{\aproblem{I}}^{\star} \leq a$, then, by \cref{thm:EfficientApproximationAlgorithm}, we guarantee that with probability at least $1 - \delta = 2 / 3$, we have
  \begin{equation}
    \abs{\hat{f}_{\aproblem{I}}^{\star} - f_{\aproblem{I}}^{\star}} 
      \leq \epsilon
    \implies 
    \hat{f}_{\aproblem{I}}^{\star} 
      \leq f_{\aproblem{I}}^{\star} + \epsilon
      \leq a + \frac{b - a}{3}
      < \frac{a + b}{2}.
  \end{equation}
  On the other hand, if $f_{\aproblem{I}}^{\star} \geq b$, then, by \cref{thm:EfficientApproximationAlgorithm}, we guarantee that with probability at least $2 / 3$, we have
  \begin{equation}
    \abs{\hat{f}_{\aproblem{I}}^{\star} - f_{\aproblem{I}}^{\star}} 
      \leq \epsilon
    \implies 
    \hat{f}_{\aproblem{I}}^{\star} 
      \geq f_{\aproblem{I}}^{\star} - \epsilon
      \geq b - \frac{b - a}{3}
      > \frac{a + b}{2}.
  \end{equation}
  So, in either case, the machine outputs the correct answer with probability at least $2 / 3$.

  \paragraph{Step 3: Technical note on the bit-length of the queried points.} As a minor technical note, \cref{alg:EfficientApproximationAlgorithm} queries the quantum oracle for $\aproblem{I}$ at $\theta \in \reals^{M}$ in a sampling grid of polynomially many equispaced points along each parameter $\theta_{j}$.
  Since these points are not guaranteed to have finite polynomial-bit encodings, they are not directly compatible with \cref{ass:QuantumQueryRoutineRealization}, which requires $\theta$ to have polynomial bit-length.
  However, since the \gls{DFT} is Lipschitz~\citep{schatzman_1996_accuracy}, we can absorb the overhead by quantizing this grid in the choice of $\epsilon$.
  To see this, choose a quantization step $\eta = \epsilon / q\xpar[\big]{\abs{\aproblem{I}}}$, for polynomial $q$ to be fixed, so that each queried parameter $\tilde{\theta}$ has polynomial bit-length in $1 / \epsilon$ and $\abs{\aproblem{I}}$. 
  Then standard \gls{DFT}/\gls{FFT} perturbation bounds (see, e.g., \citet{schatzman_1996_accuracy,gentleman_1966_fast}) imply that the induced reconstruction error is at most
  \begin{equation}
    r\xpar[\big]{\abs{\aproblem{I}}}\, \eta
    \qquad
    \text{for appropriate polynomial $r$}. 
  \end{equation}
  Therefore, by choosing $q \equiv 3 r$ we can absorb the induced quantization error into the additive error budget $\epsilon / 3$.

\end{proof}

The final piece of machinery we need to address the complexity question is the notion of polynomial-time Karp reduction, which is a standard notion of reduction in classical complexity theory~\citep{arora_2009_computational}.
For promise problems, the definition is as follows.

\begin{definition}[Polynomial-time Karp reduction]
  Consider two promise problems $\aset{P}$ and $\aset{Q}$.
  We say $\aset{P}$ is Karp-reducible to $\aset{Q}$ if there exists a deterministic polynomial-time classical procedure that maps each instance of $\aset{P}$ to an instance of $\aset{Q}$ such that a YES (resp. NO) instance of $\aset{P}$ is mapped to a YES (resp. NO) instance of $\aset{Q}$.
\end{definition}

It is well-known that the \gls{GAPMAXCUT}, i.e., the class of promise problems induced by unweighted \textsc{Max-Cut} problems, is \gls{NP}\nobreakdash-hard under polynomial-time Karp reductions.
This means that every problem in \gls{NP} is polynomial-time Karp reducible to \gls{GAPMAXCUT}, and therefore \gls{GAPMAXCUT} is at least as hard as any problem in \gls{NP}.
Thus, we can now prove the following \lcnamecref{thm:PQCOExpressivity}.

\begin{theorem}
\label{thm:PQCOExpressivity}
  Suppose \gls{GAPMAXCUT} is polynomial-time Karp reducible to \gls{GAPPQCO}. 
  Then $\gls{NP} \subseteq \gls{BQP}$, i.e., every problem in \gls{NP} is solvable by a quantum Turing machine with bounded error.
  Equivalently, if $\gls{NP} \nsubseteq \gls{BQP}$, then there does not exist a polynomial-time Karp reduction from \gls{GAPMAXCUT} to \gls{GAPPQCO}, i.e., there cannot exist a single polynomial-time deterministic mapping from all instances of \gls{GAPMAXCUT} to instances of \gls{GAPPQCO} that maps YES (resp. NO) instances of \gls{GAPMAXCUT} to YES (resp. NO) instances of \gls{GAPPQCO}.
\end{theorem}
\begin{proof}
  By \cref{thm:GAPPQCOinBPPBQP}, we have that $\gls{GAPPQCO}$ is in $\gls{BPPBQP}$.
  Therefore, if \gls{GAPMAXCUT} is polynomial-time Karp reducible to \gls{GAPPQCO}, then $\gls{GAPMAXCUT} \in \gls{BPPBQP}$.
  Since \gls{GAPMAXCUT} is \gls{NP}-hard under polynomial-time Karp reductions, every problem in $\gls{NP}$ polynomial-time Karp reduces to \gls{GAPMAXCUT}.
  By composing this reduction with the reduction from \gls{GAPMAXCUT} to \gls{GAPPQCO}, it follows that every problem in \gls{NP} belongs to \gls{BPPBQP}, and therefore in \gls{BQP}.
  Hence, $\gls{NP} \subseteq \gls{BQP}$.
\end{proof}

Since most assumptions in our framework are standard, other than $M = \bigoh(1)$, \cref{thm:PQCOExpressivity} essentially implies a concrete complexity-theoretic barrier.
Unless one expects an implausible collapse such as $\gls{NP} \subseteq \gls{BPPBQP}$, arbitrary \gls{NP}\nobreakdash-hard optimization problems should not be expected to be solvable using \gls{PQC} with a constant number of trainable parameters.
In this sense, constant-parameter \glspl{PQC} form a tractable special case rather than a universal optimization model.
In conjunction with the results of \citet{bittel_2021_training}, which show that \gls{PQC} training is, in general, \gls{NP}\nobreakdash-hard, this suggests that the barrier to tractability lies in the number of trainable parameters.
Taken together, these results suggest that the problem is tractable for $M = \bigoh(1)$, and non-tractable for $M = \bigoh(n)$, while the case of $M = \bigoh(\log n)$ remains unknown.

\section{Discussion}
\label{sec:Discussion}

This work studies \gls{PQC} optimization for poly-depth constant-parameter ansätze under standard assumptions.
The main contribution is \pgls{FPRAS} for this problem, which is the first such result to our knowledge.
The \gls{FPRAS} consists of a three-step classical--quantum pipeline consisting of a data-acquisition phase where queries to the quantum device are performed, a reconstruction phase where the optimization problem is formulated by using multidimensional \gls{FFT} on the queried data, and a classical optimization phase where the resulting \gls{HTP} optimization problem is solved by a max-degree trigonometric moment/\gls{SOS} hierarchy.
The key insight is that the optimization problem can be solved efficiently by using the structure of the generator spectra and the parameter-sharing matrix, which allows for a low-degree trigonometric moment/\gls{SOS} relaxation to be exact and for the optimization problem to be solved in polynomial time.

This result leads to several practical implications.
First, \cref{alg:EfficientApproximationAlgorithm} can always be used (under natural assumptions) for convergence to a global $\epsilon$-approximation of the optimal value with high probability, which is a stronger guarantee than what is typically provided by \gls{VQA} optimizers, which are local descent methods and can get stuck in local minima.
Second, given some \gls{PQC} $\amatrix{U}$ with integer spectral differences, possibly after rescaling, and with observable $\amatrix{O}$ fixed, our result can be used to bound the number of queries and the time complexity for reaching any target error $\epsilon$ and success probability $1 - \delta$.
This is because, for fixed $\amatrix{U}$ and $\amatrix{O}$, the assumptions are trivially satisfied at the instance level, and the complexity of solving that \gls{PQC} optimization problem can be computed directly from \cref{thm:EfficientApproximationAlgorithm}.
Finally, as the method is non-adaptive and all queries to the quantum hardware are performed in bulk before the optimization step, the total number of queries is known a priori, which allows for scheduling access to the quantum device and for the classical optimization task to run separately or be outsourced.
Taken together, these practical implications provide a useful tool for computing the difficulty of optimizing \glspl{PQC}.


\subsection*{Disclaimer}

This paper was prepared for information purposes and is not a product
of HSBC Bank Plc. or its affiliates. Neither HSBC Bank Plc. nor any
of its affiliates make any explicit or implied representation or warranty
and none of them accept any liability in connection with this paper,
including, but not limited to, the completeness, accuracy, reliability
of information contained herein and the potential legal, compliance,
tax or accounting effects thereof.

\clearpage
\appendix
\numberwithin{equation}{section}
\setcounter{equation}{0}
\renewcommand{\theequation}{\thesection\arabic{equation}}%

\clearpage
\bibliography{bibliography/bibliography-project}

\end{document}